\RequirePackage{etoolbox}
\documentclass[ba]{imsart}
\usepackage{float}
\usepackage{amsmath}
\usepackage{xcolor}
\usepackage{amsmath,amssymb,amsthm,textcomp}
\usepackage{natbib}
\usepackage[colorlinks,citecolor=blue,urlcolor=blue,filecolor=blue,backref=page]{hyperref}
\usepackage{graphicx}
\usepackage{multirow}
\usepackage{dsfont}

\startlocaldefs
\numberwithin{equation}{section}
\theoremstyle{plain}
\newtheorem{thm}{Theorem}[section]
\newtheorem{lem}{Lemma}[section]
\newtheorem{rem}{Remark}[section]
\endlocaldefs

\newcommand{\Be}{\textsf{Beta}}

\newcommand{\N}{\textsf{N}}
\newcommand{\Ber}{\textsf{Ber}}

\newcommand{\Gam}{\textsf{Gamma}}

\newcommand{\Un}{\textsf{Unif}}

\newcommand{\bbR}{\mathbb{R}}

\newcommand{\1}{\mathbf{1}_n}

\newcommand{\simiid}{\ensuremath{\mathrel{\mathop{\sim}\limits^{\rm
iid}}}}

\newcommand{\I}{\mathbf{I}}
\newcommand{\Y}{\mathbf{Y}}

\def\X{\mathbf{X}}

\def\x{\mathbf{x}}

\def\y{\mathbf{y}}

\def\Z{\mathbf{Z}}

\newcommand{\C}{\mathbf{C}}
\def\b{\boldsymbol{\beta}}

\def\g{\boldsymbol{\gamma}}

\def\d{\text{d}}
\def\Sigmab{\boldsymbol{\Sigma}}

\def\0{\noindent \textbf{0}}
\def\1{\noindent \textbf{1}}
\def\2{\noindent \textbf{2}}
\def\3{\noindent \textbf{3}}
\def\4{\noindent \textbf{4}}
\def\5{\noindent \textbf{5}}
\def\6{\noindent \textbf{6}}
\def\7{\noindent \textbf{7}}
\def\8{\noindent \textbf{8}}
\def\9{\noindent \textbf{9}}
\newcommand{\tog}{\textrm{tog}}
\newcommand{\eqdef}{=\joinrel=}

\begin{document}

\begin{frontmatter}
\title{Paired-move multiple-try stochastic search for Bayesian variable selection%\thanksref{T1}
}
\runtitle{}
%\thankstext{T1}{Footnote to the title with the ``thankstext'' command.}

\begin{aug}
\author{\fnms{Xu} \snm{Chen}\thanksref{addr1}\ead[label=e1]{xu.chen2@duke.edu}},
\author{\fnms{Shaan} \snm{Qamar}\thanksref{addr2}\ead[label=e2]{siqamar@gmail.com}}
\and
\author{\fnms{Surya} \snm{Tokdar}\thanksref{addr1}
\ead[label=e3]{tokdar@stat.duke.edu}
\ead[label=u1,url]{http://www.foo.com}}

%\runauthor{F. Author et al.}

\address[addr1]{Department of Statistical Science, Duke University, Durham, NC
    \printead{e1} % print email address of "e1"
    \printead*{e3}
}

\address[addr2]{Google Inc., Mountain View, CA
    \printead{e2}
    %\printead{u1}
}    
%\address[addr3]{Department of Statistical Science, Duke University, Durham, NC
%    \printead{e3}
    %\printead{u1}
%}

%\thankstext{t1}{Some comment}
%\thankstext{t2}{First supporter of the project}
%\thankstext{t3}{Second supporter of the project}

\end{aug}

\begin{abstract}
Variable selection is a key issue when analyzing high-dimensional data. The explosion of data with large sample sizes and dimensionality brings new challenges to this problem in both inference accuracy and computational complexity. To alleviate these problems, we propose a new scalable Markov chain Monte Carlo (MCMC) sampling algorithm for ``large $p$ small $n$'' scenarios by generalizing multiple-try Metropolis to discrete model spaces and further incorporating neighborhood-based stochastic search. The proof of reversibility of the proposed MCMC algorithm is provided. Extensive simulation studies are performed to examine the efficiency of the new algorithm compared with existing methods. A real data example is provided to illustrate the prediction performances of the new algorithm.
\end{abstract}

%\begin{keyword}[class=MSC]
%\kwd[Primary ]{60K35}
%\kwd{60K35}
%\kwd[; secondary ]{60K35}
%\end{keyword}


\end{frontmatter}

\section{Introduction}
The failure of maximum likelihood estimation in the high-dimensional $p>n$ setting naturally gives rise to the variable selection task. When predictors are known to have lower dimensional structure or it is known that only a small subset of predictors are predictive of the response, exploiting such structure can lead to dramatic improvements in the statistical efficiency of the learning  algorithm. Classic  stepwise procedures based on likelihood ratio tests for nested models or penalized model scores (e.g., Akaike's information criterion (AIC) or the Bayesian information criterion (BIC)) are generally unreliable in high dimensions. Modern Bayesian approaches to variable selection in the regression setting are typically divided into two groups: exact predictor inclusion-exclusion via spike-and-slab priors, and continuous shrinkage priors which mimic the former. With the number of possible models growing as $2^p$, direct enumeration of all models is intractable for $p\geqslant30$. While the former comes equipped with natural measures of uncertainty, such as the posterior probability of each visited model and marginal predictor inclusion probabilities, the latter often leads to more tractable inferential procedures in terms of posterior computation.

Variable selection has received a tremendous amount of attention in frequentist literature, with numerous regularization approaches enjoying much success. Most well known methods, including the \textit{Lasso} [\cite{tib1996}], \textit{SCAD} [\cite{fan2001}], \textit{adaptive lasso} [\cite{zou2006}], and the \textit{Dantzig selector} [\cite{candes2007}] are one-stage procedures, focusing on simultaneous selection and estimation of unknown model parameters; in fact, many of these come with  appealing oracle properties and asymptotic guarantees. There is an equally overwhelming body of work in the Bayesian variable selection and model averaging literature  dating back to Zellner's $g$-prior [\cite{zellner1986}]. Since then, a populous set of shrinkage priors have been developed along similar lines. Recent and notable among these include the Bayesian Lasso [\cite{park2008}], Horseshoe [\cite{polson2012}], Generalized Double Pareto \cite{armagan2013}, and Dirichlet-Laplace [\cite{bhattacharya2015}]. However, proponents of two-stages procedures, \textit{SIS} [\cite{fan2010}] and \textit{VANISH} [\cite{radchenko2010}] for example, argue that simultaneous selection and estimation is often too ambitious, instead proposing efficient variable screening algorithms which promise retaining the true support in the generalized linear model setting with high probability under regularity conditions on the design matrix. Projection pursuit regression [\cite{friedman1981}], likelihood basis pursuit [\cite{zhang2004}], and the leaps-and-bounds algorithm [\cite{furnival2000}; \cite{hoeting1999}; \cite{brusco2011}] are classic approaches to selection that utilize various optimization methods including tabu search and the branch-and-bound algorithm. Relying on penalized likelihood scoring, these methods can be effective model selection tools in simple model settings but offer no uncertainty quantification.

The spike-and-slab approach to variable selection has been predominantly developed in the linear regression setting, largely due to analytical tractability [\cite{george1993}; \cite{geweke1996}; \cite{draper1995}; \cite{carlin1995}; \cite{clyde1996}; \cite{hoeting1999}]. Here, analytical expressions for the marginal likelihood enable efficient stochastic search over the model space. The $\textrm{MC}^3$ algorithm [\cite{raftery1997}] and stochastic search variable selection (\textit{SSVS}) [\cite{george1993}] are two early Markov chain samplers that enable variable selection. \textit{SSVS} traverses the model space by scanning over the $p$-variates successively, allowing each predictor to have its state flipped, confining the search for important predictors to a local neighborhood of size $p$ at every MCMC iteration. Sequential scanning is conceptually simple, but tends to be slow as the predictor dimension grows and can suffer from mixing problems in correlated predictor settings. To mitigate this computational problem, \cite{rovckova2014} adopts EM algorithm to deterministically move toward the posterior modes instead of stochastic search. Several other stochastic search procedures have been proposed in various contexts, including applications to Gaussian graphical models and social networks [\cite{jones2005}; \cite{scott2008}], with a focus on enumerating models having high posterior probability. The authors argue that in the enormous model space, the Metropolis criterion is ``less useful as an MCMC transition kernel, and far more useful as a search heuristic for finding and cataloguing good models,'' and reliable computation of model probabilities based on frequency of occurrence in a Monte Carlo seems dubious. Shotgun stochastic search (\textit{SSS}) \cite{hans2007} proposes a neighborhood search procedure to quickly identify inclusion vectors with large posterior mass in high dimensions, and is demonstrated to perform well in linear regression and graphical model setting of moderate dimension. \cite{clyde1996} and \cite{clyde2004} discuss various Bayesian variable selection strategies for model averaging, taking advantage of specific model structure such as orthogonalized design to obtain closed-form posterior model probabilities. In addition, \cite{clyde2011} propose Bayesian adaptive sampling (\textit{BAS}) to sequentially learn marginal inclusion probabilities using a without replacement sampling algorithm. \textit{BAS} improves best over the baseline MCMC competitors when the algorithm has access to warm starts for the marginal predictor inclusion probabilities. \cite{berger2005} propose a stochastic search algorithm that incorporates local proposals which explore a neighborhood around a catalogue of previously sampled models using initial estimates for posterior model and predictor inclusion probabilities to guide the traversal. Using a path-sampling approach to efficiently compute Bayes factors between one-away pairs of models in the linear regression setting, their strategy yields a connected graph between all explored models with the hope that this set is large enough to reliably estimate approximate inferential statistics.

Acknowledging the tension between local efficiency and mode-finding while maintaining MCMC reversibility, we adapt concepts from neighborhood-based stochastic search \cite{hans2007} and generalize the multiple-try Metropolis (MTM) \cite{liu2000} algorithm for efficient sampling of inclusion vectors. Key innovations address the challenges faced by variable selection sampling in high dimensions; in particular, a scalable MCMC sampler should
\begin{enumerate}
\item effectively trade-off exploration and exploitation: in high dimensions, MCMC should employ a mechanism for adaptation as means of efficiently discovering regions of high posterior probability;
\item have an efficient transitioning scheme: poor mixing can result when good proposals are identified but rejected because of the reversibility constraint; and
\item cut the computational budget: when likelihood evaluations are expensive, rejections are wasteful. A flexible sampling scheme will allow for a host of proposals, allowing for an annealing process toward modal regions of the posterior.
\end{enumerate}
The rest of the paper is organized as follows. In Section \ref{sec2}, we start by establishing notations and presenting the hierarchical formulation for Bayesian variable selection with conjugate priors for regression coefficients and the predictor inclusion vectors which are adopted in simulation studies. Then we briefly review the shotgun stochastic search and multiple-try Metropolis algorithms. We propose a new scalable MCMC sampler for predictor inclusion vectors by generalizing the multiple-try Metropolis and combining with neighborhood-based stochastic search in Section \ref{sec3}. Extensive simulation studies are provided in Section \ref{sec4} to examine the effectiveness of the proposed algorithm according to inference accuracy, prediction performances and computational efficiency. We conclude in Section \ref{sec5} with discussions on the future research directions.
\section{Background}
\label{sec2}
\subsection{Bayesian variable selection}
Consider the canonical Gaussian linear regression model
\[\Y=\X\b+\mathbf{\varepsilon},\qquad\mathbf{\varepsilon}\sim \N(0,\I_n/\phi)\]
where $\Y\in \bbR^n$ is a response vector and $\X\in \mathbb{R}^{n\times p}$ is a design matrix for $n$ samples and $p$ predictors. Assume $\Y$ and $\X$ are standardized and hence an intercept is not included. $\b\in \mathbb{R}^p$ is an unknown regression coefficient vector. Accurately recovering the support of and estimating $\b$ are of interest when $p$ is large. 

Under Bayesian scheme, variable selection is typically performed by introducing a $p\times 1$ binary latent indicator vector $\g\in \{0,1\}^p$. Denote the set of indices of predictors $\{1,2,...,p\}$ as $[p]$. For each $i\in [p]$, $\gamma_i=1$ if $\X_i$ is included in the model. $\g$ can also be viewed as the set of indices of active predictors (i.e., a subset of $[p]$) in the affiliated model $\mathcal{M}_{\g}$ [\cite{yang2015}] where $|\g|$ and $\g^c$ denote the cardinality and complement of $\g$. Under $\mathcal{M}_{\g}$, a conjugate hierarchical model is typically constructed as follows:
\begin{align}
\b_{\g}\mid\g,\phi\sim \N(\0, \Sigmab_{\g}/\phi)
\end{align}
An independent prior is obtained by specifying $\Sigmab_{\g}=\I_{|\g|}$. Another conventional choice is a $g$-prior where $\Sigmab_{\g}=g(\X_{\g}^T\X_{\g})^{-1}$ [\cite{zellner1986}]. This type of prior preserves the correlation structure of a design matrix and leads to simple closed-form marginal likelihoods. Models with small sizes are preferred when larger values of $g$ are adopted. See \cite{liang2012} for a detailed discussion of the effects of $g$.
\begin{align}
\phi\sim \Gam(a, b)
\end{align}
Generally, $a$ and $b$ are chosen to be small constants, resulting in a non-informative prior for $\phi$. However, $\phi$ is expected to be larger when including more predictors in the model. Therefore, \cite{george1993} and \cite{dobra2004} consider relating $a$ or $b$ to $|\g|$. When $a, b\rightarrow 0$, we get a popular improper prior $\pi(\phi)\propto 1/\phi$.
\begin{align}
\pi(\g\mid \tau)=\prod_{j=1}^p\tau^{\gamma_j}(1-\tau)^{1-\gamma_j}=\tau^{|\g|}(1-\tau)^{p-|\g|}
\end{align}
The prior for $\g$ only depends on its size. Fixing $\tau=1/2$ yields a uniform distribution for all $2^p$ models with expected model size of $p/2$. This prior fails to penalize large models. A more reasonable approach is to treat $\tau$ as a hyperparameter with Beta prior. See \cite{scott2010} for theoretical properties of this prior. 
\begin{align}
\tau\sim \Be(u,v)
\end{align}
Let $d^*$ be the number of expected model size. We may set $u=d^*$ and $v=p-d^*$ resulting in $\mathbf{E}[\tau]=d^*$ and $\mathbf{Var}[\tau]\approx d^*/p^2$ when $d^*=o(p)$. A marginal Beta-binomial distribution for $\g$ is 
\begin{align}
\pi(\g)=\frac{B(|\g|+u, p-|\g|+v)}{B(u, v)}
\end{align} 
where $B(\cdot,\cdot)$ is the Beta function.

All the simulations performed in this paper adopt $g$-prior with $g=n$ (i.e., Unit information prior [\cite{kass1995}]), $\pi(\phi)\propto 1/\phi$ and a Beta-binomial prior for $\g$. Under these settings, the marginal likelihood is given by
\begin{align}
\mathcal{L}_n(\Y\mid \g)&=\int \pi(\Y \mid \b_{\g},\phi)\pi(\b_{\g}\mid\phi,\g)\pi(\phi)\d\b_{\g}\d\phi\\
&=\frac{\Gamma(n/2)(1+g)^{n/2}}{\pi^{n/2}\|\Y\|_2^n}\frac{(1+g)^{-|\g|/2}}{[1+g(1-R_{\g}^2)]^{n/2}}
\end{align}
where $\Gamma(\cdot)$ is the Gamma function and $R_{\g}^2$ is the ordinary coefficient of determination for the model $\mathcal{M}_{\g}$
\begin{align}
R_{\g}^2=\frac{\Y^T\mathbf{P}_{\X_{\g}}\Y}{\|\Y\|_2^2}
\end{align}
with $\mathbf{P}_{\X_{\g}}=\X_{\g}(\X_{\g}^T\X_{\g})^{-1}\X_{\g}^T$ the projection matrix onto the column space of $\X_{\g}$.
\subsection{Neighborhood-based stochastic search MCMC samplers}
Let $T(\g,\cdot)$ be a proposal transition function over $N(\g)$, the neighborhood set of $\g$. Here $T(\g,\g')>0\Longleftrightarrow T(\g',\g)>0$ is required to guarantee reversibility. Then a Metropolis-Hastings (MH) random walk neighborhood search algorithm is implemented iteratively as follows:
\begin{enumerate}
\item Randomly select a proposal state $\g'\in N(\g)$ according to $T(\g,\cdot)$. 
\item Accept proposal $\g'$ with probability $\alpha$ where
\[\alpha(\g,\g')=\min\bigg\{1,\frac{\pi(\g'\mid\Y)T(\g',\g)}{\pi(\g\mid\Y)T(\g,\g')}\bigg\}\]
otherwise stay at $\g$.
\end{enumerate}
This algorithm generates an irreducible, aperiodic, and positive recurrent Markov chain. Let $1_j$ be a $p\times 1$ vector with $j^{th}$ element 1 and others 0. Then neighborhood set considered by \cite{hans2007} consists of three types of moves:
\begin{enumerate}
\item Add an inactive predictor: $N_A(\g)=\{\g'\mid \g'=\g+1_j, j\in\g^c\}$
\item Remove an active predictor: $N_R(\g)=\{\g'\mid \g'=\g-1_j, j\in\g\}$
\item Swap an active predictor with an inactive predictor: $N_S(\g)=\{\g'\mid \g'=\g-1_j+1_k, (j,k)\in\g\times\g^c\}$
\end{enumerate}
and $N(\g)=N_A(\g)\cup N_R(\g)\cup N_S(\g)$. Swap move is actually the combination of adding and removing. \cite{yang2015} further unifies 1 and 2 into one class based on Hamming distance. 

A MCMC sampler built on SSS provided in \cite{hans2007} can be obtained immediately by defining 
\begin{align}
T(\g,\g')=\frac{S(\g')}{\sum_{\tilde{\g}\in N(\g)}S(\tilde{\g})}
\label{equ2.9}
\end{align} 
where $S$ is any positive score function.
\subsection{Multiple-try Metropolis}
Multiple-try Metropolis algorithm is proposed by \cite{liu2000} to mitigate the potential slow convergence problem of traditional MH algorithms. Instead of only considering a single proposal, MTM proposes multiple trials each iteration to prevent the chain from being stuck in local modes in a continuous state space. Specifically, suppose $\pi$ is the target distribution and $T$ is a transition kernel. Further define weight $\omega(\x,\y)=\pi(\x)T(\x,\y)$. Then a general MTM algorithm involve the following procedures:
\begin{enumerate}
\item Sample $M$ i.i.d. proposals $\y_1,\y_2,...,\y_M$ according to $T(\x,)$. 
\item Select $\y\in\{\y_1,\y_2,...,\y_M\}$ with probability proportional to $\omega(\y_j,\x)\quad j=1,2,...,M$. 
\item Sample backward set $\{\x_1^*,\x_2^*,...,\x_{M-1}^*\}$ according to $T(\y,)$ and set $\x_M^*=\x$.
\item Accept the proposal $\y$ with probability $\alpha$ where
\[\alpha(\x,\y)=\min\bigg\{1,\frac{\sum_{j=1}^M \omega(\y_j,\x)}{\sum_{j=1}^M \omega(\x_j^*,\y)}\bigg\}\]
otherwise stay at $\x$.
\end{enumerate}
This algorithm generates a reversible Markov chain leaving $\pi$ as the invariant distribution. The standard MH sampler results as a special case when $M=1$. In \cite{liu2000}, MTM is demonstrated to be more efficient on multimodal state space exploration than traditional MH algorithms through simulation studies. \cite{pandolfi2010} extends this approach by further incorporating an additional weight $\omega^*(\x,\y)$. The original MTM is obtained when $\omega^*(\x,\y)=\omega(\x,\y)$.
\section{A paired-move multiple-try stochastic search sampler}
\label{sec3}
The MCMC sampler built on SSS may suffer from two problems. First, the chain may be stuck due to substantially low acceptance rates. Suppose that the current state is $\g$ and $\g'$ is proposed. When $\g'$ has better neighborhoods than $\g$, $\sum_{\tilde{\g}\in N(\g')}S(\tilde{\g})$ is much larger than $\sum_{\tilde{\g}\in N(\g)}S(\tilde{\g})$ leading to a small acceptance rate. Therefore, the sampler may be able to identify inclusion vectors with high posterior probabilities but fail to transition to them. Another concern is computational complexity. We notice that $|\g|+|\g'|$ remove, $2p-|\g|-|\g'|$ add, and $|\g|(p-|\g|)+|\g'|(p-|\g'|)$ swap neighborhoods are evaluated in each iteration. This $O(p)$ cost is further exacerbated when $n$ is large. Although likelihood scores can be evaluated in parallel, most one-away neighborhoods offer little improvement to the model fit in high dimensions. Reducing the lengtha of a chain is inevitable when computational budget is limited, resulting in poor mixing and unreliable inferences. 

We propose a new MCMC sampler by combining the idea of neighborhood-based stochastic search and MTM to address the issues described above. Specifically, a paired-move strategy is introduced in Section \ref{3.1} to improve acceptance rates. In Section \ref{3.2}, multiple-try scheme is generalized to discrete model spaces to allow for a flexible and efficient neighborhood search. We further incorporate adaptive scores for predictors according to the correlation structure of a design matrix and previous posterior samples to improve mixing in Section \ref{3.3}.
\subsection{Paired-move neighborhood search}
\label{3.1}
The paired-move strategy is motivated by the following fact:
\begin{align}
\g'\in N_A(\g)&\Longleftrightarrow \g\in N_R(\g')\nonumber\\ 
\g'\in N_R(\g)&\Longleftrightarrow \g\in N_A(\g')\\
\g'\in N_S(\g)&\Longleftrightarrow \g\in N_S(\g')\nonumber
\end{align}
Therefore, a forward move $\g\rightarrow\g'$ and a corresponding backward move $\g'\rightarrow\g$ are paired. We proposed a paired-move reversible neighborhood sampler (\textit{pRNS}) with ``add-remove'', ``remove-add'', and ``swap-swap'' forward-backward neighborhoods. By allowing different moves to be proposed separately and in efficient succession, \textit{pRNS} can dramatically improve mixing in the space of single predictor changes to $\g$. The \textit{pRNS} proposal transition function is defined by
\begin{align}
T(\g,\g')=w_AT_A(\g,\g')+w_RT_R(\g,\g')+w_ST_S(\g,\g')
\label{move-type}
\end{align}
where $w_A$, $w_R$, and $w_S$ are probabilities of proposing add, remove, and swap moves respectively and $T_A$, $T_R$, and $T_S$ are proposal transition functions as in \ref{equ2.9} restricted to their corresponding sets of neighborhoods. 

Naturally, $w_A$, $w_R$, and $w_S$ are positive and sum to 1. These probabilities are allowed to vary with the current model size to encourage additions to smaller models, removal from larger models, and swaps for moderately sized ones. As a general rule of configurations, $w_A(|\g|)$ can be specified to be monotone decreasing with respect to $|\g|$ with $w_A(0)=w_R(p)=1$ and $w_S(|\g|)>0$ when $0<|\g|<p$. Moreover, we recommend adopting a unimodal $w_S(|\g|)$ with a mode near $d^*$ and a ``light tail'': $\sum_{d^*<|\g|<p}w_S(|\g|)<\delta$ where $\delta=0.1$, for example. Note that when $d^*=o(p)$, random-walk Gibbs samplers are heavily biased toward attempting adding additional predictors instead of removing undesirable ones. The inefficiency of random selection is addressed by utilizing the suggested rules. For simplicity, the following settings are adopted in all simulations in the paper:
\begin{align}
w_A(0)=w_R(p)=1 \quad\textrm{ and }\quad w_A(|\g|)=w_R(|\g|)=w_S(|\g|)=\frac{1}{3}\hspace*{1.8mm}\textrm{if $0<|\g|<p$}
\end{align}
The resulting MCMC algorithm adopting \textit{pRNS} is as follows:
\begin{enumerate}
\item Select move $m\in\{A,R,S\}$ with probabilities $w_A,w_R$, and $w_S$.
\item Construct the forward set of neighborhoods $N_m(\g)$.
\item Randomly select a proposal state $\g'\in N_m(\g)$ according to $T_m(\g,\cdot)$.
\item Construct the backward set of neighborhoods $N_m'(\g')$ where $m'\in\{R,A,S\}$ is the backward move correspond to $m$.
\item Accept the proposal $\g'$ with probability $\alpha$ where
\begin{align}
\alpha(\g,\g')=\min\bigg\{1,\frac{\pi(\g'\mid\Y)[w_{m'}(|\g'|)T_{m'}(\g',\g)]}{\pi(\g\mid\Y)[w_{m}(|\g|)T_{m}(\g,\g')]}\bigg\}
\label{acc-rej-pRNS}
\end{align}
otherwise stay at $\g$.
\end{enumerate}
\begin{lem}
The paired-move reversible neighborhood sampler (\textit{pRNS}) with acceptance probability \ref{acc-rej-pRNS} satisfies the detailed balance condition leaving the desired target distribution $\pi(\g\mid\Y)$ invariant. 
\end{lem}
\begin{proof}
Let $A(\g_1,\g_2)$ be the actual transition probability for moving from $\g_1$ to $\g_2$. Then, we have
\begin{align}
\pi(\g\mid\Y)A(\g,\g')&=\pi(\g\mid\Y)w_m(|\g|)T_m(\g,\g')\min\bigg\{1,\frac{\pi(\g'\mid\Y)[w_{m'}(|\g'|)T_{m'}(\g',\g)]}{\pi(\g\mid\Y)[w_{m}(|\g|)T_{m}(\g,\g')]}\bigg\}\nonumber\\
&=\min\left\{\pi(\g\mid\Y)w_m(|\g|)T_m(\g,\g'),\pi(\g'\mid\Y)[w_{m'}(|\g'|)T_{m'}(\g',\g)]\right\}
\label{db-pRNS}
\end{align}
Note that the expression \ref{db-pRNS} is symmetric in $\g$ and $\g'$ and hence $\pi(\g\mid\Y)A(\g,\g')=\pi(\g'\mid\Y)A(\g',\g)$ which is the detailed balance condition.
\end{proof}
\begin{rem}
Since the neighborhoods evaluated in each iteration are restricted to a subset of $N(|\g|)$, \textit{pRNS} efficiently reduces the computational cost, though add and swap neighborhoods remain $O(p)$. As the dimension of predictors $p$ grows, an additional mechanism is essential to limit the size of neighborhoods which is the main concern of Section \ref{3.2}.
\end{rem}
\subsection{A paired-move multiple-try stochastic search MCMC algorithm}
\label{3.2}
It is inefficient to evaluate a large number of neighborhoods in each iteration. A flexible computational cost is desired to accommodate for the requirement of inference accuracy and the computational budget. One attractiveness of the MTM is that the computational cost can be adjusted by tuning the number of trails $M$.
\subsubsection{A mixed discrete multiple-try sampler}
We adapt MTM to the discrete model space where transitions are confined to the neighborhoods of inclusion vector $\g$. Instead of considering all neighborhoods, we propose a general framework for generating a stochastic set of neighborhoods of the current state $\g$. To formulate our method, we first define the ``toggle function'' $\tog: [p]\times \{0,1\}^p\rightarrow \{0,1\}^p$ as follows:
\begin{align}
\tog(i,\g=(\gamma_1,\gamma_2,...,\gamma_i,...,\gamma_p))=(\gamma_1,\gamma_2,...,1-\gamma_i,...,\gamma_p)
\end{align}
Namely, if $i^{th}$ predictor is included(excluded) in the current state $\g$, then $\g'=\tog(i,\g)$ is a neighborhood removing(adding) $i^{th}$ predictor. Note that a swap move  between $j^{th}$ and $k^{th}$ predictors is $\tog(k,\tog(j,\g))$ (or $\tog(j,\tog(k,\g))$) for $\gamma_j+\gamma_k=1$. 

To introduce stochasticity, we further define $\eta_i\sim\Ber(\omega(\gamma_i,v_i))$ for $i\in[p]$ with a weight function $\omega: \{0,1\}\times \mathbb{R}^+\rightarrow [0,1]$ taking inputs $\gamma_i$ and a nonnegative predictor importance score $v_i$. For simplicity, we do not consider swap moves which will be handled in detail in the next section and focus on a mixed set of neighborhoods only containing add and remove neighborhoods for now. Under these settings, the forward set of neighborhoods of $\g$ is defined as $N_{mix}(\g)=\{\tog(i,\g)\mid\eta_i=1,i\in[p]\}$ and $T_{mix}(\g,\cdot)$ is a proposal transition function as in \ref{equ2.9} restricted to $N_{mix}(\g)$. An algorithm for this generalized discrete MTM (dMTM) over a model space is:
\begin{enumerate}
\item For current state $\g$ and $i\in[p]$, independently sample $\eta_i\sim \Ber(\omega(\gamma_i,v_i))$.
\item Form the forward mixed set of neighborhoods of $\g$: $N_{mix}(\g)=\{\tog(i,\g)\mid\eta_i=1,i\in[p]\}$.
\item Select $\g'=\tog(i^*,\g)\in N_{mix}(\g)$ according to $T_{mix}(\g,\cdot)$.
\item For the proposed state $\g'$ and $j\neq i^*\in[p]$, independently sample $\eta_j'\sim \Ber(\omega(\gamma_j',v_j))$ and set $\eta_{i^*}'=1$.
\item Form the backward mixed set of neighborhoods of $\g'$: $N_{mix}'(\g')=\{\tog(j,\g')\mid\eta_j'=1,j\in[p]\}$.
\item Accept the proposal $\g'$ with probability $\alpha$ where
\begin{align}
\alpha(\g,\g')=\min\bigg\{1,\frac{\pi(\g'\mid \Y)[\omega(\gamma_{i^*}',v_{i^*})T'_{mix}(\g',\g)]}{\pi(\g\mid \Y)[\omega(\gamma_{i^*},v_{i^*})T_{mix}(\g,\g')]}\bigg\}
\label{acc-rej-dMTM}
\end{align}
otherwise stay at $\g$.
\end{enumerate}
\begin{lem}
The discrete MTM (dMTM) algorithm with acceptance probability \ref{acc-rej-dMTM} satisfies the detailed balance condition leaving the desired target distribution $\pi(\g\mid\Y)$ invariant. 
\label{lem3.2}
\end{lem}
\begin{proof}
We will prove this lemma together with Theorem \ref{thm3.1} in the next section.
\end{proof}
\begin{rem}
Efficient strategies of specifying $v_i$ and $\omega_i$ for $i\in[p]$ would enhance the possibilities of including important predictors and excluding undesirable ones. An adaptive configuration is provided in Section \ref{3.3}. 
\end{rem}
\subsubsection{A paired-move multiple-try stochastic search sampler}
\label{pMTM}
A paired-move multiple-try stochastic search MCMC algorithm (\textit{pMTM}) is obtained as a special case of the dMTM algorithm under the following configuration of weight function $\omega$:
\begin{align}
\omega(\gamma_i,v_i;m)=(1-\gamma_i)f(v_i)\mathds{1}_{\{m=A\}}+\gamma_ig(v_i)\mathds{1}_{\{m=R\}}
\label{wt-fun}
\end{align}
where $\mathds{1}_{\{\cdot\}}$ is an indicator function and $f,g: \mathbb{R}^+\rightarrow [0,1]$ determining the probabilities of including and removing predictors. Note that here we further take the type of move into account. It is reasonable because it allows for including an important predictor $i$ with high probability (large $f(v_i)$) and being preserved (small $g(v_i)$). Then the \textit{pMTM} algorithm is given as:
\begin{enumerate}
\item Select move $m\in\{A,R,S\}$ with probabilities $w_A(|\g|),w_R(|\g|)$, and $w_S(|\g|)$.
\item (a) If move $m\in\{A,R\}$: for $i\in[p]$, independently sample $\eta_i\sim\Ber(\omega(\gamma_i,v_i;m))$. Define the forward add or remove set as $N_F(\g)=\{\tog(i,\g)\mid\eta_i=1,i\in[p]\}$.\smallskip\\
(b) If move $m=S$: for $(a,r)\in\g^c\times\g$, sample $\eta_a\sim\Ber(\omega(\gamma_a,v_a;A))$, and independently sample $\eta_r\sim\Ber(\omega(\gamma_r,v_r;R))$ (totally sample $ar$ Bernoulli random variables). Define the forward swap set as $N_F(\g)=\{\tog(a,\tog(r,\g))\mid\eta_a=\eta_r=1,(a,r)\in\g^c\times\g\}$.
\item Select $\g'\in N_{F}(\g)$ according to $T_{F}(\g,\cdot)$. If $m\in\{A,R\}$, denote $\g'=\tog(i^*,\g)$; otherwise denote $\g'=\tog(a^*,\tog(r^*,\g))$ for $m=S$.
\item (a) If move $m=A$: for $j\neq i^*$, sample $\eta_j'\sim\Ber(\omega(\gamma_j',v_j;R))$ and set $\eta_{i^*}'=1$. Define the backward remove set as $N_B(\g')=\{\tog(j,\g')\mid\eta_j'=1,j\in[p]\}$.\smallskip\\
(b) If move $m=R$: for $j\neq i^*$, sample $\eta_j'\sim\Ber(\omega(\gamma_j',v_j;A))$ and set $\eta_{i^*}'=1$. Define the backward add set as $N_B(\g')=\{\tog(j,\g')\mid\eta_j'=1,j\in[p]\}$.\smallskip\\
(c) If move $m=S$: for $(a',r')\in(\g')^c\times\g'$, sample $\eta_{a'}'\sim\Ber(\omega(\gamma_{a'}',v_{a'};A))$, and independently sample $\eta_{r'}'\sim\Ber(\omega(\gamma_{r'}',v_{r'};R))$ (totally sample $a'r'$ Bernoulli random variables) and set $(\eta_{a^*}',\eta_{r^*}')=(1,1)$. Define the backward swap set as $N_B(\g')=\{\tog(a',\tog(r',\g'))\mid\eta_{a'}'=\eta_{r'}'=1,(a',r')\in(\g')^c\times\g'\}$.
\item For $m\in\{A,R,S\}$, the corresponding backward paired-move is $m'\in\{R,A,S\}$. Accept the proposal $\g'$ with probability $\alpha$ where
\begin{align}
\alpha(\g,\g')=\left\{\begin{array}{ll}
\min\bigg\{1,\frac{\pi(\g'\mid \Y)[w_{m'}(|\g'|)\omega(\gamma_{i^*}',v_{i^*};m')T'_{B}(\g',\g)]}{\pi(\g\mid \Y)[w_{m}(|\g|)\omega(\gamma_{i^*},v_{i^*};m)T_{F}(\g,\g')]}\bigg\} & \textrm{if $m\in\{A,R\}$}\medskip\\
\min\bigg\{1,\frac{\pi(\g'\mid \Y)[\omega(\gamma_{r^*}',v_{r^*};A)\omega(\gamma_{a^*}',v_{a^*};R)T'_{B}(\g',\g)]}{\pi(\g\mid \Y)[\omega(\gamma_{a^*},v_{a^*};A)\omega(\gamma_{r^*},v_{r^*};R)T_{F}(\g,\g')]}\bigg\} & \textrm{if $m=S$}
\end{array}\right.
\label{acc-rej-pMTM}
\end{align}
otherwise stay at $\g$.
\end{enumerate}
\begin{thm}
The paired-move multiple-try stochastic search MCMC (\textit{pMTM}) algorithm with acceptance probability \ref{acc-rej-pMTM} satisfies the detailed balance condition leaving the desired target distribution $\pi(\g\mid\Y)$ invariant. 
\label{thm3.1}
\end{thm}
\begin{proof}
Let $A(\g_1,\g_2)$ be the actual transition probability for moving from $\g_1$ to $\g_2$.\smallskip\\
If $m\in\{A,R\}$: Let $\omega_i=\omega(\gamma_i,v_i;m)$ and $\tilde{\omega}_i=\omega(\gamma_i',v_i;m')$ denote the probabilities of the forward and backward move for predictor $i\in[p]$. Then for $\g'=\tog(i^*,\g)$,
\begin{align}
&\pi(\g\mid\Y)A(\g,\g')\nonumber\\
&=\pi(\g\mid\Y)w_m(|\g|)\sum\limits_{\substack{\eta,\eta'\in\{0,1\}^p\\ \eta_{i^*}=\eta_{i^*}'=1}}\bigg[\omega_{i^*}\bigg\{\prod\limits_{j\neq i^*}\omega_j^{\eta_j}(1-\omega_j)^{1-\eta_j}\tilde{\omega}_j^{\eta_j'}(1-\tilde{\omega}_j)^{1-\eta_j'}\bigg\}\nonumber\\
&\times\left. T_F(\g,\g')\min\left\{1,\frac{\pi(\g'\mid \Y)[w_{m'}(|\g'|)\tilde{\omega}_{i^*}T'_{B}(\g',\g)]}{\pi(\g\mid \Y)[w_{m}(|\g|)\omega_{i^*}T_{F}(\g,\g')]}\right\}\right]\nonumber\\
&=\sum\limits_{\substack{\eta,\eta'\in\{0,1\}^p\\ \eta_{i^*}=\eta_{i^*}'=1}}\bigg[\bigg\{\prod\limits_{j\neq i^*}\omega_j^{\eta_j}(1-\omega_j)^{1-\eta_j}\tilde{\omega}_j^{\eta_j'}(1-\tilde{\omega}_j)^{1-\eta_j'}\bigg\}\nonumber\\
&\times\min\left\{\pi(\g\mid \Y)[w_{m}(|\g|)\omega_{i^*}T_{F}(\g,\g')],\pi(\g'\mid \Y)[w_{m'}(|\g'|)\tilde{\omega}_{i^*}T'_{B}(\g',\g)]\right\}\bigg]
\label{db-pMTM1}
\end{align}
If $m=S$: Note that a swap move can be viewed as a composition of a remove and an add move. Denote the probability of a forward move for a pair of predictors $(a,r)\in \g^c\times\g$ as $\omega_a\omega_r$ where $\omega_a=\omega(\gamma_a,v_a;A)$ and $\omega_r=\omega(\gamma_r,v_r;R)$. Likewise for backward move probabilities, let $\tilde{\omega}_{a'}=\omega(\gamma_{a'}',v_{a'};A)$ and $\tilde{\omega}_{r'}=\omega(\gamma_{r'}',v_{r'};R)$ for $(a',r')\in (\g')^c\times\g'$. Note that $(\ast): w_S(|\g|)=w_S(|\g'|)$ since $|\g|=|\g'|$. Then for $\g'=\tog(a^*,\tog(r^*,\g))$, we have
\begin{align}
&\pi(\g\mid\Y)A(\g,\g')\nonumber\\
&=\pi(\g\mid\Y)w_S(|\g|)\sum\limits_{\substack{(a,r)\in\g^c\times\g\\ (a',r')\in(\g')^c\times\g'\\ \eta_{a^*}=\eta_{r^*}=1\\ \eta_{a^*}'=\eta_{r^*}'=1}}\bigg[\omega_{a^*}\omega_{r^*}\bigg\{\prod\limits_{(a,r)\neq(a^*,r^*)}\omega_a^{\eta_a}(1-\omega_a)^{1-\eta_a}\omega_r^{\eta_r}(1-\omega_r)^{1-\eta_r}\nonumber\\
&\times\prod\limits_{(a',r')\neq(a^*,r^*)}\tilde{\omega}_{a'}^{\eta_{a'}'}(1-\tilde{\omega}_{a'})^{1-\eta_{a'}'}\tilde{\omega}_{r'}^{\eta_{r'}'}(1-\tilde{\omega}_{r'})^{1-\eta_{r'}'}\bigg\}\nonumber\\
&\times\left. T_F(\g,\g')\min\bigg\{1,\frac{\pi(\g'\mid \Y)[\tilde{\omega}_{r^*}\tilde{\omega}_{a^*}T'_{B}(\g',\g)]}{\pi(\g\mid \Y)[\omega_{a^*}\omega_{r^*}T_{F}(\g,\g')]}\bigg\}\right]\nonumber\\
&\stackrel{(\ast)}{\eqdef}\sum\limits_{\substack{(a,r)\in\g^c\times\g\\ (a',r')\in(\g')^c\times\g'\\ \eta_{a^*}=\eta_{r^*}=1\\ \eta_{a^*}'=\eta_{r^*}'=1}}\bigg[\bigg\{\prod\limits_{(a,r)\neq(a^*,r^*)}\omega_a^{\eta_a}(1-\omega_a)^{1-\eta_a}\omega_r^{\eta_r}(1-\omega_r)^{1-\eta_r}\nonumber\\
&\times\prod\limits_{(a',r')\neq(a^*,r^*)}\tilde{\omega}_{a'}^{\eta_{a'}'}(1-\tilde{\omega}_{a'})^{1-\eta_{a'}'}\tilde{\omega}_{r'}^{\eta_{r'}'}(1-\tilde{\omega}_{r'})^{1-\eta_{r'}'}\bigg\}\nonumber\\
&\times\min\left\{w_S(|\g|)\pi(\g\mid \Y)[\omega_{a^*}\omega_{r^*}T_{F}(\g,\g')],w_S(|\g'|)\pi(\g'\mid \Y)[\tilde{\omega}_{r^*}\tilde{\omega}_{a^*}T'_{B}(\g',\g)]\right\}\bigg]
\label{db-pMTM2}
\end{align}
Note that the expressions \ref{db-pMTM1} and \ref{db-pMTM2} are symmetric in $\g$ and $\g'$ and hence $\pi(\g\mid\Y)A(\g,\g')=\pi(\g'\mid\Y)A(\g',\g)$ which is the detailed balance condition.
\end{proof}
\begin{rem}
The proof is established on a general form of $T$ as in \ref{equ2.9}. If we specify $S(\tilde{\g})\propto\pi(\tilde{\g}\mid\Y)$, the unnormalized marginal posterior probability for $\tilde{\g}$, then the acceptance ratio $\alpha$ is 
\begin{align}
\alpha(\g,\g')=\left\{\begin{array}{ll}
\min\bigg\{1,\frac{w_{m'}(|\g'|)\omega(\gamma_{i^*}',v_{i^*};m')\sum_{\tilde{\g}\in N_F(\g)}\pi(\tilde{\g}\mid\Y)}{w_{m}(|\g|)\omega(\gamma_{i^*},v_{i^*};m)\sum_{\tilde{\g}\in N_B(\g')}\pi(\tilde{\g}\mid\Y)}\bigg\} & \textrm{if $m\in\{A,R\}$}\medskip\\
\min\bigg\{1,\frac{\omega(\gamma_{r^*}',v_{r^*};A)\omega(\gamma_{a^*}',v_{a^*};R)\sum_{\tilde{\g}\in N_F(\g)}\pi(\tilde{\g}\mid\Y)}{\omega(\gamma_{a^*},v_{a^*};A)\omega(\gamma_{r^*},v_{r^*};R)\sum_{\tilde{\g}\in N_B(\g')}\pi(\tilde{\g}\mid\Y)}\bigg\} & \textrm{if $m=S$}
\end{array}\right.
\label{acc-rej-pMTM-s}
\end{align}
All simulations in the paper are performed under this setting. Note that $\pi(\tilde{\g}\mid\Y)\propto\mathcal{L}(\Y\mid\tilde{\g})\pi(\tilde{\g})$. When the sample size $n$ is large, computing $\mathcal{L}(\Y\mid\tilde{\g})$ will be expensive. An alternative choice is using Laplace approximation of the marginal likelihood $\hat{\mathcal{L}}(\Y\mid\tilde{\g})$ and hence $S(\tilde{\g})=\hat{\mathcal{L}}(\Y\mid\tilde{\g})\pi(\tilde{\g})$.
\end{rem}
\begin{rem}
This non-trivial generalization of the MTM algorithm extends MCMC for sampling high-dimensional inclusion vectors. The framework is general and flexible, allowing for varied settings based on different problems, structures of datasets and computational budgets. Adopting adaptive importance scores for predictors within this framework is discussed in the next section.  
\end{rem}
\subsection{Adaptive predictor importance}
\label{3.3}
Weight function $\omega(\gamma_i,v_i;m)$ for $i\in[p]$ and $m\in\{A,R\}$ provides a mechanism to improve mixing and robustness for sampling predictor inclusion vectors in both low-signal and high-dimensional settings. In spectrometry or gene expression data, for example, predictors are often highly correlated because of their spatial proximity, and therefore ``exchangeable" in the sense of their explanatory power. It is well known that penalized methods such as the \textit{Lasso} [\cite{tib1996}] often simply selects one out of a set of highly correlated predictors, and the \textit{elastic net} penalty [\cite{zou2005}] is often a more robust shrinkage method to use in such settings. A regularized estimate for the correlation matrix [\cite{schafer2005}; \cite{bickel2008}] or other similarity measures between predictors may be used to efficiently update importance scores.  

Suppose that $f$ and $g$ in \ref{wt-fun} are monotone increasing and decreasing functions of $v_i$ respectively. Therefore, as importance scores are updated, predictors with large $v_i$ are promoted within add neighborhoods and demoted in remove neighborhoods. Denote the length of the MCMC chain as $T$ with a burnin period $b_0$. Define a $p\times p$ thresholded absolute correlation matrix $\C$ with
\begin{align}
C_{ij}=|\rho_{ij}|\mathds{1}_{\{|\rho_{ij}|>\varepsilon\}}
\label{abs-cor}
\end{align} 
where $\rho_{ij}=\textrm{Cor}(\X_i,\X_j)$ is the empirical correlation between predictors $i,j\in[p]$ and $\varepsilon\in(0,1)$ is a pre-specified threshold. By incorporating the correlation structure of the design matrix and the history of the MCMC chain, we introduce an adaptive importance scores for predictors. For for all $i\in[p]$ at $(t+1)^{th}$ iteration, $v_i(t+1)$ is updated as follows:
\begin{align}
v_i(t+1)=v_i(t)+z(i,\g)\left(\frac{t}{b_0}\mathds{1}_{\{t\leqslant b_0\}}+\frac{1}{(t-b_0)^{\zeta}}\mathds{1}_{\{t>b_0\}}\right)
\label{update-score}
\end{align}
with $z:[p]\times\{0,1\}^p\rightarrow[0,1]$; in particular, $z(i,\g)=(1-\gamma_i)(\sum_{j=1}^{p}\gamma_jC_{ij}/\sum_{j=1}^{p}\gamma_j)+\gamma_i$. We suggest specifying the learning rate $\zeta\in(0.5,1]$ as $2/3$ following convention from stochastic gradient descent. Further modification may be adopting the quantile of $|\rho_{ij}|$s as the threshold to ensure a fixed ratio of entries of $\C$ are zeros.

Based on this updating scheme for importance scores, we propose an adaptive version of the \textit{pMTM} where the probability of the $i^{th}$ predictor to be included in an add (remove) set of neighborhoods is (inversely) proportional to its importance score $v_i$. Specifically, we define
\begin{align}
f(v_i)=\frac{Mv_i}{Mv_i+p}\qquad g(v_i)=\frac{1}{v_i}
\end{align}
Under this configuration, we suggest initializing the \textit{pMTM} sampler with $\g=(0,0,...,0)^T$ and $v_i=1$ for all $i\in[p]$. Accordingly, $M$ can be viewed as a target ``neighborhood budget'' noting that the expected number of add neighborhoods is $\sum_{i\notin\g}f(v_i)\approx M(p-|\g|)/(p+M)\approx M$ initially when $M=o(p)$. When the true model size $d=o(p)$, most of the importance scores retain $v_i\approx 1$ and hence the stochastic control of the number of neighborhoods is maintained. Stationarity of the \textit{pMTM} sampler is preserved subject to diminishing adaptation of predictor importance scores [\cite{roberts2007}]. This adaptive version of the \textit{pMTM} sampler is denoted as \textit{ada-pMTM}.
\section{Numerical studies}
\label{sec4}
In this section\footnote{All the simulations are run in \texttt{R} on a computer with x86$\times$64 Intel(R) Core(TM) i7-3770k.}, two examples are provided to illustrate the effectiveness of \textit{pMTM} and \textit{ada-pMTM} on model space exploration. Through intensive simulation studies, the comparisons between \textit{pMTM} and various frequentist and Bayesian variable selection methods exhibit the state-of-the-art performance of our framework. An analysis of a real data example adopting \textit{pMTM} is presented to demonstrate the use of the proposed algorithms. We close this section by a comparison of computational efficiency on a toy example. Except for Section \ref{comeff}, all simulations are performed without parallelization.

We first specify the tunning parameters adopted in all simulations in this section. Burnin period $b_0$ is set be the first $20\%$ of the total length of the chain. We adopt the updating scheme \ref{update-score} with $\zeta=2/3$ and $\C$ with the $75\%$ quantile of $|\rho_{ij}|$s as the threshold. For simplicity, $g(v_i)$ for all $i\in[p]$ is specified as 1 which means that all remove neighborhoods are included in the forward move set when a remove move is proposed.
\subsection{Effectiveness of pMTM}
%\subsubsection{Accumulated posterior mass}
%pSTM, pSTMadpt, pMTM and pMTMadpt are implemented on a synthetic dataset to compare their efficiencies on exploration of the model space. The synthetic data is based on a gene expression dataset considered by [Hans et. al., 2007] where the data description and an initial analysis can be found in [Rich et. al., 2005]. We specify $n=50$ and $p=8000$ and the synthetic data set is generated by further incorporating a revised version of example 2 in [Nott and Green, 2004]. Specifically, the synthetic dataset is simulated as follows: We generate $Z,Z_1,Z_2,...,Z_{10}\simiid\N(0,I_n)$ and then set $X_{i}=\rho_1 Z+2Z_i$ for $i=1,3,5,8,9,10$, $X_{i}=\rho_2 X_{i-1}+ \rho_3 Z_{i}$ for $i=2,4,6$, $X_7=\rho_4(X_8+X_9-X_{10})+\rho_5 Z_7$. $(X_{11},X_{12},X_{13},X_{14},X_{15})\sim N(0,\Sigma)$ where $\Sigma$ is the covariance matrix presented in the supplementary material of [Rich et. al., 2005]. And then $\beta$ is specified by $\beta_\gamma=(1.5,1,-1.5,-1,-0.8,0.9,0.8,1)$ for $\gamma=(1,3,5,7,8,11,12,13)$.

%All algorithms are run for approximately $10^9$ marginal likelihood evaluations. The accumulated posterior mass is evaluated based on the top $10^6$ models in each set of posterior samples.

%\textbf{(Interpretation of figure 1)}
%\begin{figure} 
%\includegraphics[width=8cm]{1.pdf}
%\caption[]{Accumulated Posterior Mass for pMTM and pSTM \textbf{(Need to be updated)}}
%\label{fig1}
%\end{figure}
In this section, we compare the proposed algorithms with two traditional Gibbs samplers, random-scan Gibbs and systematic-scan Gibbs, based on their efficiencies on exploration of the model space. \cite{george1993} and \cite{george1997} describe a systematic-scan Gibbs sampler by sequentially updating components of $\g$ according to $\pi(\gamma_i\mid \gamma_{-i},\Y)$ for $i\in[p]$ in one iteration. A random-scan Gibbs sampler will randomly select an index $i$ first and then update the corresponding $\gamma_i$

The algorithms are compared using a simulated dataset based on the number of marginal likelihood evaluations needed to find the true model. Simulated data is based on the dataset used in \cite{west2001} which contain 49 patients and each of which has gene expression data including 3883 genes. In terms of the rank of contributions of different genes to tumor provided in the supporting information 3 in \cite{west2001}, we extract \texttt{TFF1} (rank 1), \texttt{ESR1} (rank 2), \texttt{CYP2B6} (rank 3) and \texttt{IGFBP2} (rank 5) to form the true predictors. The reason why we didn't choose the $4^{th}$ gene \texttt{TFF3} is that it has a high correlation with \texttt{TFF1}.

The simulated dataset is constructed as follows: we first normalized these four genes and further combined with standard multivariate normals to form the design matrix. Then $\b$ is specified by $\b_{\g}=(1.3,0.3,-1.2,-0.5)$ for $\g=\{1,2,3,4\}$ with a sequence of increasing values of $p=50,100,150,...,500$. $\varepsilon$ is standard normal with mean 0 and variance 0.5. The values of regression coefficients and variance of noise are exactly same with the example in Section 4.4 of \cite{hans2007}. Hyperparameters are specified as $u=4$, $v=p-4$, $M=p/10$. We report the median of results based on 100 synthetic dataset for each value of $p$ in Figure.\ref{fig2}. Circles and crosses are employed to represent the value larger than $\log(3\times 10^5)$.

\begin{figure}[!htbp]
\centering
\includegraphics[width=10cm]{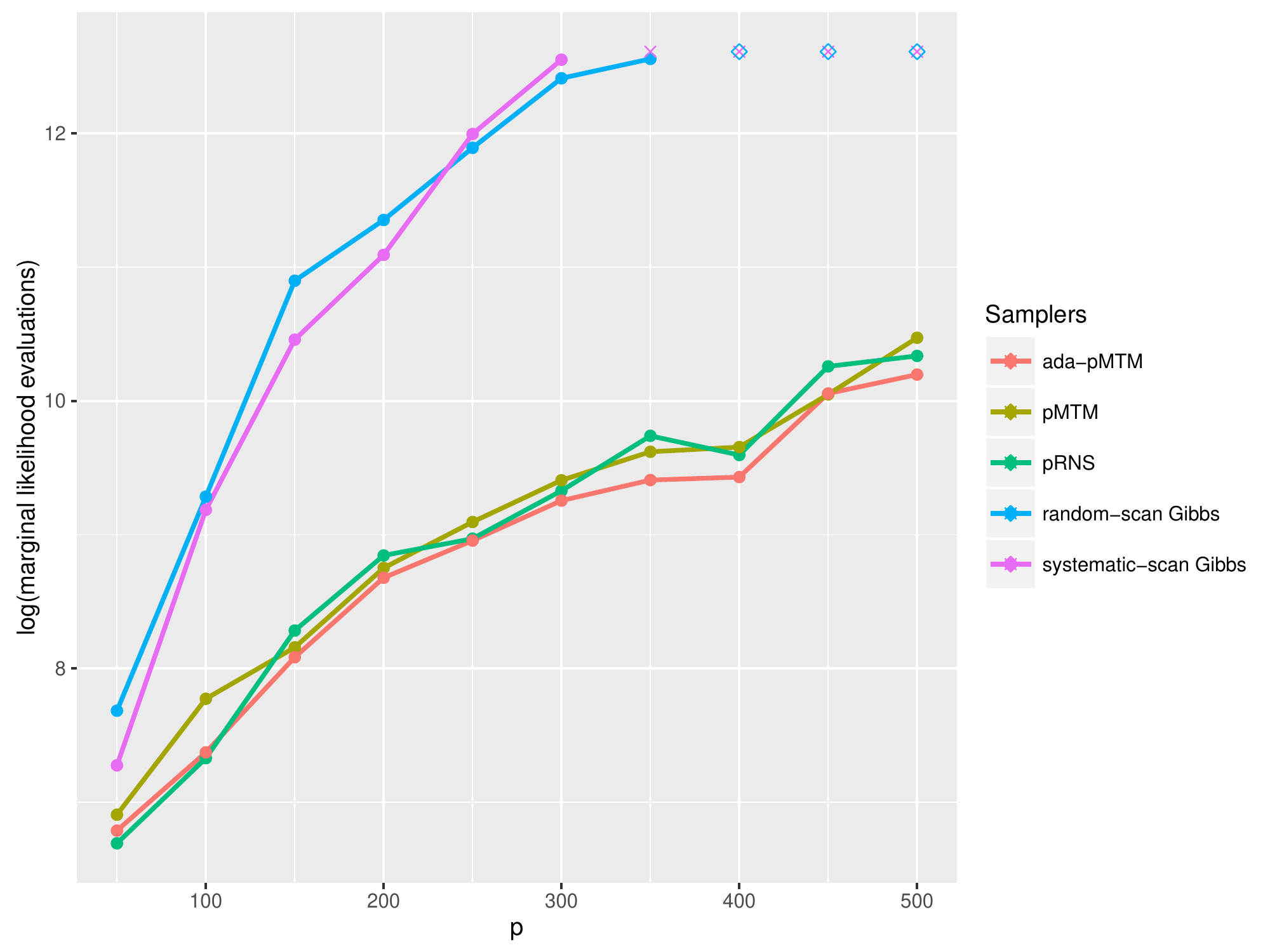}
\caption[]{Logarithm of the median number of marginal likelihood evaluations needed to find the true model}
\label{fig2}
\end{figure}
\begin{figure}[!htbp]
\centering
\includegraphics[width=10cm]{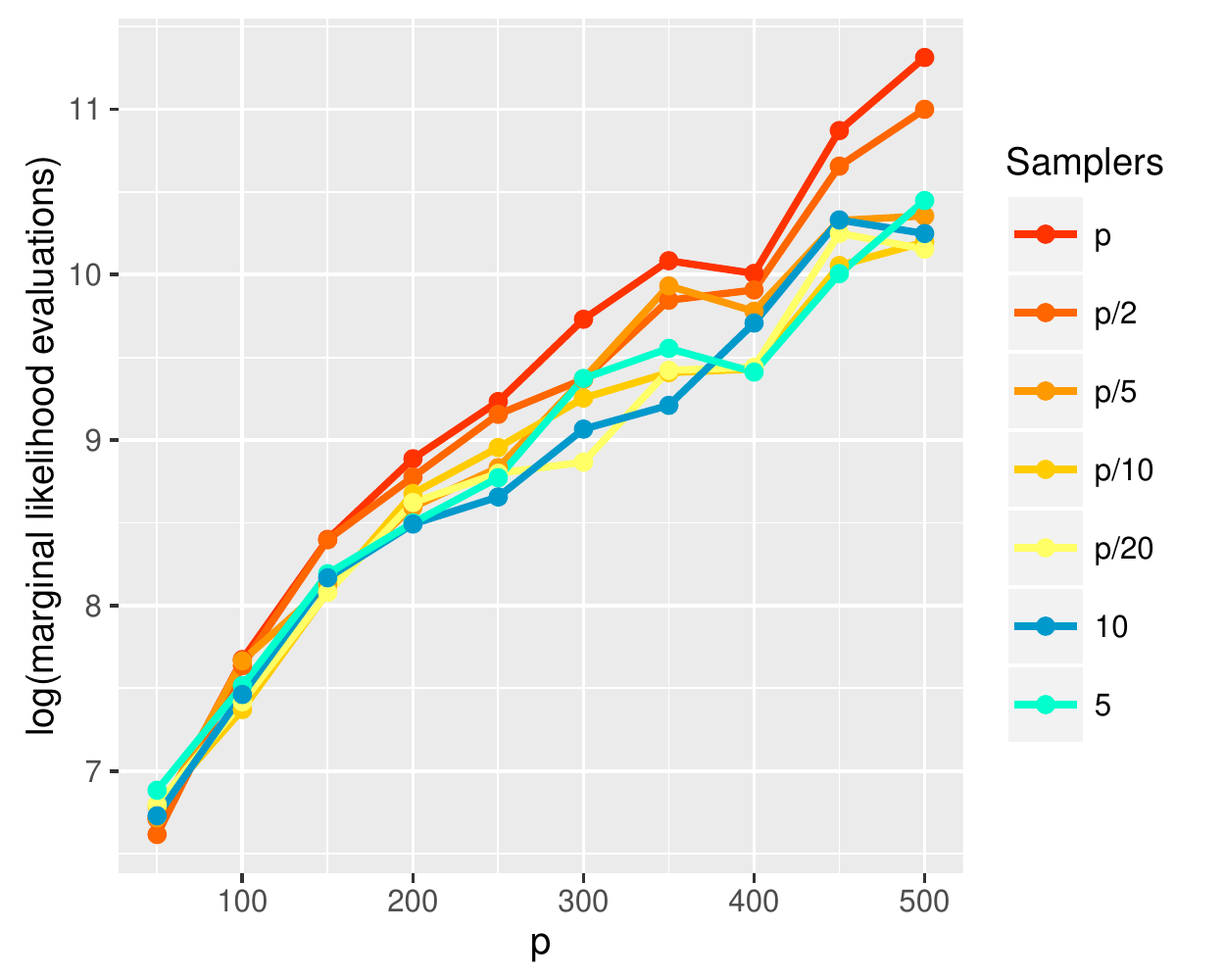}
\caption[]{Logarithm of the median number of marginal likelihood evaluations needed to find the true model for \textit{ada-pMTM} with different expected number of trails ($M$).}
\label{fig5}
\end{figure}

%The curve for pSTMadpt is not included in Figure 2 because it takes much longer time to find the true model. For example, even when $p=50$, pSTMadpt will typically need to evaluate more than $10^6(\approx e^{14})$ marginal likelihoods until the true model is detected. pSTMadpt can be easily stuck in local modes since it only consider one proposal each iteration with increasing weights of predictors existed in previous posterior samples.
 
According to the graph, it is clear that the paired-move strategy can effectively reduce the computational cost comparing to two Gibbs samplers. When $p=50,100$, \textit{ada-pMTM} is slightly worse than \textit{pRNS} but still competitive. As $p$ becomes larger, \textit{ada-pMTM} dominates all other algorithms which verifies the claim that multiple trials can help the sampler move out of local modes.\\
\\
To explore how the choice of $M$ influence the efficiency of \textit{ada-pMTM}, we further implemented \textit{ada-pMTM} with different choices of $M$. Specifically, two groups of $M$ are specified as:
\begin{itemize}
\item a function of $p$: $M=M(p)=p,p/2,p/5,p/10,p/20$, and
\item a fixed value: $M=5,10$
\end{itemize}
As displayed in Figure.\ref{fig5}, a vague pattern appears. It should be clear that the numbers of marginal likelihood evaluation needed for $M=p$ or $p/2$ are larger suggesting the less efficiency for large $M$ while the other 5 choices do not show significant differences. We will use $M=p/10$ throughout the rest of the simulation studies.
\subsection{Simulated examples}
\label{4.2}
Intensive simulation studies are presented for assessing the performances of proposed algorithms. In particular, we compare proposed algorithms to random- and systematic- scan Gibbs [\cite{george1993}; \cite{george1997}], \textit{EMVS} [\cite{rovckova2014}], \textit{Lasso} [\cite{tib1996}], \textit{adaptive lasso} [\cite{zou2006}] and \textit{SCAD} [\cite{fan2001}]. Different structures and degrees (low, moderate and high) of correlation of the design matrix are considered. Specifically, we specify $(n,p)=(100,1000)$ for all experiments. The signal-to-noise ratio $\|\b\|_2/\sigma$ is adjusted to guarantee that fitting results for different methods are moderate and comparable. Four structures are described as follows:
\begin{enumerate}
\item \textbf{Independent design:}
This example is originally analyzed by \cite{fan2008} with $t=4,5$. We further make this example difficult by setting $t=1$. 
\begin{align*}
\X_1,\X_2,...,\X_p&\simiid\N(0,\I_n)\\
\beta_i&=(-1)^{U_i}(t\log n/\sqrt{n}+|N(0,1)|)\text{ }\text{where $U_i\simiid\Un(0,1)$ }\text{for $i=1,2,....,8$}\\
\varepsilon&\sim N(0,1.5^2\I_n)
\end{align*}
\item \textbf{Compound symmetry:} This example is revised based on the Example 1 in \cite{fan2008} where $\|\b\|_2$ is much smaller here. Every pair of predictors has the same theoretical correlation $\rho$. We adopt $\rho=0.3,0.6$ and $0.9$ to allow for different degrees of correlation.
\begin{align*}
\X_1,\X_2,...,\X_p&\simiid\N(0,\Sigmab) \quad\text{with $\Sigma_{ij}=\rho$ for $i\neq j$ and 1 for $i=j$}\\
\b&=(2.0,2.5,-2.0,2.5,-2.5,0,...,0)\\
\varepsilon&\sim N(0,1.5^2\I_n)\\
\end{align*}
\item \textbf{Autoregression:} This example is modified from the Example 2 in \cite{tib1996}. This type of correlation structure widely exists in time series. Again, we set $\rho=0.3,0.6,0.9$. 
\begin{align*}
\X_1,\X_2,...,\X_p&\simiid\N(0,\Sigmab) \quad\text{with $\Sigma_{ij}=\rho^{|i-j|}$}\\
\g&=\{1,2,3,4,5,20,35,60,90,150,151,300\}\\
\b_{\g}&=(2,-3,2,2,-3,3,-2,3,-2,3,2,-2)\\
\varepsilon&\sim N(0,2^2\I_n)
\end{align*}
\item \textbf{Group structure:} This example is revised from the simulated experiment 2 in \cite{bottolo2010} and first analyzed in \cite{nott2005}. A group structured correlation exhibits: collinearity exists between $\X_i$ and $\X_{i+1}$ for $i=1,3,5$ and linear relationship is presented in group $(\X_7,\X_8,\X_9,\X_{10})$ and $(\X_{11},\X_{12},\X_{13},\X_{14},\X_{15})$. Whether the algorithm can select the correct predictors and do not select variables from the
second block are of interest.
The model is simulated as follows: 
\begin{align*}
\Z,\Z_1,\Z_2,...,\Z_{15}&\simiid\N(0,\I_n)\\
\X_{i}&=\rho_1 \Z+2\Z_i \quad\text{for $i=1,3,5,8,9,10,12,13,14,15$}\\
\X_{i}&=\rho_2 \X_{i-1}+ \rho_3 \Z_{i}\quad\text{for $i=2,4,6$}\\
\X_7&=\rho_4(\X_8+\X_9-\X_{10})+\rho_5 \Z_7\\
\X_{11}&=\rho_5(\X_{14}+\X_{15}-\X_{12}-\X_{13})+\rho_5 \Z_{11}\\
\b_{\g}&=(1.5,1.5,1.5,1.5,-1.5,1.5,1.5,1.5) \text{ with } \g=\{1,3,5,7,8,11,12,13\}
\end{align*}
\end{enumerate}
where $\rho_i\quad i=1,2,...,5$ are adjusted to import small, moderate and high correlation into each group.

\texttt{R} packages \texttt{glmnet}, \texttt{parcor}, \texttt{ncvreg} and \texttt{EMVS} are used for \textit{Lasso}, \textit{adaptive lasso}, \textit{SCAD} and \textit{EMVS} respectively. The tunning parameter $\lambda$'s are specified by minimizing cross-validation errors. All Bayesian methods are implemented with Beta-Binomial prior using same hyperparameters: $u=10$, $v=p-10$. We consider the highest probability model (HPM), median probability model (MPM) [\cite{barb2004}] (i.e., the model containing predictors with inclusion probability larger than 0.5) and Bayesian model averaging (BMA) for proposed algorithms. Recommended default settings for \textit{EMVS} are adopted except for a more elaborate sequences of $v_0$. The comparisons are based on five metrics: \textbf{Average model size:} number of predictors selected; \textbf{Runtime:} running time for different methods (fixed running time for all MCMC samplers); \textbf{FN:} number of false negatives; \textbf{FP:} number of false positives; \textbf{FDR:} false discovery rate and \textbf{$L_2$ distance:} $\|\hat{\b}-\b\|_2$. For each scenario, 100 synthetic datasets are simulated and the mean of above  metrics are reported for assessment.
\begin{table}[H]
\begin{center}
\scriptsize
\begin{tabular}{c c | c c c  c  c c}
\hline
\hline
\multicolumn{8}{c}{$p$=1000}  \\
& & \textit{Average model size} & \textit{FN} & \textit{FP} & \textit{FDR} & $\mathbf{\|\hat{\beta}-\beta\|_2}$  & \textit{Runtime}\\
\hline
\multirow{3}*{ada-pMTM}& MPM & 6.92 & 1.60 & 0.52 & 0.06545 & 0.96247 &\multirow{3}*{9.98}\\
& HPM & 7.02 & 1.65 & 0.67 & 0.07717 & 0.97281 &\\
& BMA & - & - & - & - & \color{red}\textbf{0.92419} & \\
\hline
\multirow{3}*{pMTM}& MPM & 6.80 & 1.64 & 0.44 & 0.05683 & 0.96258 &\multirow{3}*{10.03}\\
& HPM & 6.78 & 1.78 & 0.56 & 0.06033 & 0.99559 &\\
& BMA & - & - & - & - & \textbf{0.94179} &\\
\hline
\multirow{3}*{pRNS (equal evaluations)}& MPM & 7.30 & 1.53 & 0.83 & 0.09555 & 1.00338 &\multirow{3}*{37.15}\\
& HPM & 7.96 & 1.53 & 1.49 & 0.15161 & 1.07369 &\\
& BMA & - & - & - & - & 0.98907 & \\
\hline
\multirow{3}*{pRNS (equal time)}& MPM & 7.61 & 1.56 & 1.17 & 0.12756 & 1.04951 &\multirow{3}*{9.09}\\
& HPM & 7.77 & 1.61 & 1.38 & 0.14125 & 1.08198 &\\
& BMA & - & - & - & - & 1.02115 & \\
\hline
\multirow{3}*{pRNS (equal iterations)}& MPM & 2.16 & 6.07 & 0.23 & 0.06045 & 2.86641 &\multirow{3}*{0.24}\\
& HPM & 2.22 & 6.12 & 0.34 & 0.10738 & 2.91814 &\\
& BMA & - & - & - & - & 2.72615 & \\
\hline
\multicolumn{2}{c|}{EMVS}  & 4.47 & 3.56 & 0.03 & 0.00421 & 1.52944 & 13.75\\
\multicolumn{2}{c|}{Lasso} & 21.23 & 1.34 & 14.57 & 0.59485 & 1.69590 & 2.07\\
\multicolumn{2}{c|}{adaptive lasso} & 12.13 & 1.34 & 5.47 & 0.36516 & 1.12010 & 3.93\\
\multicolumn{2}{c|}{SCAD} & 33.24 & 0.36 & 25.60 & 0.75887 & 0.99172 & 6.56\\
\hline
\hline
\end{tabular}
\end{center}
\caption{Independent design with $(n,p_0)=(100,8)$}
\end{table}

\begin{table}[H]
\begin{center}
\scriptsize
\begin{tabular}{c c | c c c  c  c c}
\hline
\hline
\multicolumn{8}{c}{$\rho$=0.3}  \\
& & \textit{Average model size} & \textit{FN} & \textit{FP} & \textit{FDR} & $\mathbf{\|\hat{\beta}-\beta\|_2}$  & \textit{Runtime}\\
\hline
\multirow{3}*{ada-pMTM}& MPM & 5.15 & 0.00 & 0.15 & 0.02357 & \textbf{0.44765} &\multirow{3}*{10.74}\\
& HPM & 5.13 & 0.00 & 0.13 & 0.02119 & \color{red}\textbf{0.44244} &\\
& BMA & - & - & - & - & 0.48411 & \\
\hline
\multirow{3}*{pMTM}& MPM & 4.99 & 0.24 & 0.23 & 0.04086 & 0.92833 &\multirow{3}*{10.35}\\
& HPM & 5.18 & 0.06 & 0.24 & 0.03808 & 0.58596 &\\
& BMA & - & - & - & - & 0.88684 &\\
\hline
\multirow{3}*{pRNS (equal evaluations)}& MPM & 5.33 & 0.00 & 0.33 & 0.05167 & 0.49293 &\multirow{3}*{37.33}\\
& HPM & 5.45 & 0.00 & 0.45 & 0.06899 & 0.52963 &\\
& BMA & - & - & - & - & 0.52373 & \\
\hline
\multirow{3}*{pRNS (equal time)}& MPM & 5.57 & 0.00 & 0.57 & 0.07907 & 0.55974 &\multirow{3}*{10.87}\\
& HPM & 5.65 & 0.00 & 0.65 & 0.08597 & 0.57467 &\\
& BMA & - & - & - & - & 0.56371 & \\
\hline
\multirow{3}*{pRNS (equal iterations)}& MPM & 2.41 & 2.82 & 0.23 & 0.09117 & 3.88156 &\multirow{3}*{0.21}\\
& HPM & 2.82 & 2.98 & 0.80 & 0.23667 & 4.03775 &\\
& BMA & - & - & - & - & 3.59336 & \\
\hline
\multicolumn{2}{c|}{EMVS} & 5.03 & 0.00 & 0.03 & 0.00500 & 0.83232 & 13.47\\
\multicolumn{2}{c|}{Lasso} & 15.28 & 0.00 & 10.28 & 0.59420 & 1.52074 & 2.04\\
\multicolumn{2}{c|}{adaptive lasso} & 6.12 & 0.00 & 1.12 & 0.13605 & 0.52061 & 3.95\\
\multicolumn{2}{c|}{SCAD} & 8.64 & 0.00 & 3.64 & 0.24591 & 0.45551 & 4.25\\
\hline
\hline
\multicolumn{8}{c}{$\rho$=0.6}  \\
& & \textit{Average model size} & \textit{FN} & \textit{FP} & \textit{FDR} & $\mathbf{\|\hat{\beta}-\beta\|_2}$  & \textit{Runtime}\\
\hline
\multirow{3}*{ada-pMTM}& MPM & 5.18 & 0.01 & 0.19 & 0.02946 & 0.57993 &\multirow{3}*{10.50}\\
& HPM & 5.20 & 0.01 & 0.21 & 0.03065 & \color{red}\textbf{0.56093} &\\
& BMA & - & - & - & - & 0.63864 & \\
\hline
\multirow{3}*{pMTM}& MPM & 5.16 & 0.11 & 0.27 & 0.04093 & 0.79130 &\multirow{3}*{10.43}\\
& HPM & 5.20 & 0.04 & 0.24 & 0.03594 & 0.63270 &\\
& BMA & - & - & - & - & 0.85272 &\\
\hline
\multirow{3}*{pRNS (equal evaluations)}& MPM & 5.66 & 0.00 & 0.66 & 0.08798 & 0.72689 &\multirow{3}*{36.36}\\
& HPM & 5.71 & 0.00 & 0.71 & 0.09543 & 0.74347 &\\
& BMA & - & - & - & - & 0.74913 & \\
\hline
\multirow{3}*{pRNS (equal time)}& MPM & 5.61 & 0.00 & 0.61 & 0.08156 & 0.70428 &\multirow{3}*{9.98}\\
& HPM & 5.73 & 0.00 & 0.73 & 0.09294 & 0.72342 &\\
& BMA & - & - & - & - & 0.74506 & \\
\hline
\multirow{3}*{pRNS (equal iterations)}& MPM & 2.08 & 3.15 & 0.23 & 0.07867 & 4.17630 &\multirow{3}*{0.26}\\
& HPM & 2.55 & 3.12 & 0.67 & 0.28433 & 4.10882 &\\
& BMA & - & - & - & - & 3.79928 & \\
\hline
\multicolumn{2}{c|}{EMVS}  & 5.05 & 0.01 & 0.06 & 0.01 & 1.01934 & 11.64\\
\multicolumn{2}{c|}{Lasso} & 18.36 & 0.00 & 13.36 & 0.63903 & 2.03600 & 2.29\\
\multicolumn{2}{c|}{adaptive lasso} & 7.59 & 0.01 & 2.60 & 0.26217 & 0.80726 & 4.44\\
\multicolumn{2}{c|}{SCAD} & 6.81 & 0.00 & 1.81 & 0.15528 & 0.58029 & 4.53\\
\hline
\hline
\multicolumn{8}{c}{$\rho$=0.9}  \\
& & \textit{Average model size} & \textit{FN} & \textit{FP} & \textit{FDR} & $\mathbf{\|\hat{\beta}-\beta\|_2}$  & \textit{Runtime}\\
\hline
\multirow{3}*{ada-pMTM}& MPM & 3.38 & 1.87 & 0.25 & 0.08229 & 3.02116 &\multirow{3}*{10.22}\\
& HPM & 3.45 & 1.93 & 0.38 & 0.10679 & 3.02510 &\\
& BMA & - & - & - & - & \color{red}\textbf{2.79824} & \\
\hline
\multirow{3}*{pMTM}& MPM & 3.52 & 1.90 & 0.42 & 0.11136 & 3.15391 &\multirow{3}*{10.88}\\
& HPM & 3.58 & 1.96 & 0.54 & 0.14769 & 3.15457 &\\
& BMA & - & - & - & - & \textbf{2.95858} &\\
\hline
\multirow{3}*{pRNS (equal evaluations)}& MPM & 3.98 & 1.83 & 0.81 & 0.18968 & 3.26200 &\multirow{3}*{36.32}\\
& HPM & 4.27 & 1.86 & 1.13 & 0.24183 & 3.34397 &\\
& BMA & - & - & - & - & 3.16544 & \\
\hline
\multirow{3}*{pRNS (equal time)}& MPM & 3.86 & 1.86 & 0.72 & 0.17876 & 3.26003 &\multirow{3}*{10.57}\\
& HPM & 4.24 & 1.89 & 1.13 & 0.25938 & 3.42405 &\\
& BMA & - & - & - & - & 3.13012 & \\
\hline
\multirow{3}*{pRNS (equal iterations)}& MPM & 1.44 & 3.96 & 0.40 & 0.22317 & 4.90713 &\multirow{3}*{0.26}\\
& HPM & 1.87 & 3.95 & 0.82 & 0.37617 & 4.95334 &\\
& BMA & - & - & - & - & 4.51167 & \\
\hline
\multicolumn{2}{c|}{EMVS}  & 2.79 & 2.51 & 0.30 & 0.11138 & 3.80729 & 10.26\\
\multicolumn{2}{c|}{Lasso} & 11.54 & 1.43 & 7.97 & 0.60918 & 4.03553 & 3.76\\
\multicolumn{2}{c|}{adaptive lasso} & 9.74 & 0.98 & 5.72 & 0.48575 & 3.10494 & 7.10\\
\multicolumn{2}{c|}{SCAD} & 4.67 & 1.95 & 1.62 & 0.29208 & 3.47218 & 2.95\\
\hline
\hline
\end{tabular}
\end{center}
\caption{Compound symmetry with $(n,p,p_0)=(100,1000,5)$}
\end{table}

\begin{table}[H]
\begin{center}
\scriptsize
\begin{tabular}{c c | c c c  c  c c}
\hline
\hline
\multicolumn{8}{c}{$\rho$=0.3}  \\
& & \textit{Average model size} & \textit{FN} & \textit{FP} & \textit{FDR} & $\mathbf{\|\hat{\beta}-\beta\|_2}$  & \textit{Runtime}\\
\hline
\multirow{3}*{ada-pMTM}& MPM & 11.07 & 1.07 & 0.14 & 0.01620 & 1.75973 &\multirow{3}*{15.19}\\
& HPM & 11.36 & 1.15 & 0.51 & 0.04634 & 1.74790 &\\
& BMA & - & - & - & - & 1.81494 & \\
\hline
\multirow{3}*{pMTM}& MPM & 10.05 & 2.16 & 0.21 & 0.02417 & 2.86410 &\multirow{3}*{16.33}\\
& HPM & 10.67 & 2.32 & 0.99 & 0.09460 & 2.80901 &\\
& BMA & - & - & - & - & 2.87099 &\\
\hline
\multirow{3}*{pRNS (equal evaluations)}& MPM & 12.55 & 0.00 & 0.55 & 0.04031 & 0.95725 &\multirow{3}*{43.94}\\
& HPM & 13.08 & 0.00 & 1.08 & 0.07350 & 0.99831 &\\
& BMA & - & - & - & - & 0.96122 & \\
\hline
\multirow{3}*{pRNS (equal time)}& MPM & 11.88 & 0.21 & 0.09 & 0.01115 & 0.99528 &\multirow{3}*{14.89}\\
& HPM & 11.84 & 0.30 & 0.14 & 0.01478 & 1.05328 &\\
& BMA & - & - & - & - & 1.03719 & \\
\hline
\multirow{3}*{pRNS (equal iterations)}& MPM & 1.93 & 10.29 & 0.22 & 0.10267 & 7.90704 &\multirow{3}*{0.27}\\
& HPM & 2.18 & 10.27 & 0.45 & 0.15617 & 7.96132 &\\
& BMA & - & - & - & - & 7.50963 & \\
\hline
\multicolumn{2}{c|}{EMVS}  & 9.57 & 2.78 & 0.35 & 0.06049 & 3.07500 & 27.05\\
\multicolumn{2}{c|}{Lasso} & 42.45 & 1.03 & 31.48 & 0.71510 & 4.60301 & 2.28\\
\multicolumn{2}{c|}{adaptive lasso} & 19.55 & 0.58 & 8.13 & 0.37514 & 2.27123 & 4.27\\
\multicolumn{2}{c|}{SCAD} & 25.29 & 0.00 & 13.29 & 0.46581 & \color{red}\textbf{0.95686} & 4.68\\
\hline
\hline
\multicolumn{8}{c}{$\rho$=0.6}  \\
& & \textit{Average model size} & \textit{FN} & \textit{FP} & \textit{FDR} & $\mathbf{\|\hat{\beta}-\beta\|_2}$  & \textit{Runtime}\\
\hline
\multirow{3}*{ada-pMTM}& MPM & 9.89 & 2.56 & 0.45 & 0.05300 & 3.40834 &\multirow{3}*{13.24}\\
& HPM & 10.44 & 2.78 & 1.22 & 0.10091 & 3.60832 &\\
& BMA & - & - & - & - & 3.47301 & \\
\hline
\multirow{3}*{pMTM}& MPM & 9.02 & 3.59 & 0.61 & 0.06996 & 4.32600 &\multirow{3}*{12.58}\\
& HPM & 9.46 & 3.82 & 1.28 & 0.12107 & 4.44422 &\\
& BMA & - & - & - & - & 4.14006 &\\
\hline
\multirow{3}*{pRNS (equal evaluations)}& MPM & 12.42 & 0.16 & 0.58 & 0.04133 & 1.21687 &\multirow{3}*{41.96}\\
& HPM & 13.24 & 0.16 & 1.40 & 0.09381 & 1.31555 &\\
& BMA & - & - & - & - & 1.22476 & \\
\hline
\multirow{3}*{pRNS (equal time)}& MPM & 12.39 & 0.32 & 0.71 & 0.05421 & \textbf{1.38717} &\multirow{3}*{13.51}\\
& HPM & 13.24 & 0.32 & 1.56 & 0.10572 & 1.48584 &\\
& BMA & - & - & - & - & \color{red}\textbf{1.37363} & \\
\hline
\multirow{3}*{pRNS (equal iterations)}& MPM & 2.57 & 9.89 & 0.46 & 0.15567 & 7.88735 &\multirow{3}*{0.25}\\
& HPM & 2.78 & 9.92 & 0.70 & 0.21045 & 7.95898 &\\
& BMA & - & - & - & - & 7.53342 & \\
\hline
\multicolumn{2}{c|}{EMVS}  & 7.63 & 4.81 & 0.44 & 0.06244 & 5.06578 & 23.88\\
\multicolumn{2}{c|}{Lasso} & 30.88 & 3.36 & 22.24 & 0.65295 & 5.94437 & 2.36\\
\multicolumn{2}{c|}{adaptive lasso} & 17.81 & 3.05 & 8.86 & 0.41424 & 5.00711 & 4.42\\
\multicolumn{2}{c|}{SCAD} & 30.20 & 1.69 & 19.89 & 0.63170 & 3.64840 & 5.79\\
\hline
\hline
\multicolumn{8}{c}{$\rho$=0.9}  \\
& & \textit{Average model size} & \textit{FN} & \textit{FP} & \textit{FDR} & $\mathbf{\|\hat{\beta}-\beta\|_2}$  & \textit{Runtime}\\
\hline
\multirow{3}*{ada-pMTM}& MPM & 6.14 & 6.76 & 0.90 & 0.14867 & 6.69856 &\multirow{3}*{12.65}\\
& HPM & 7.12 & 6.74 & 1.86 & 0.22984 & 6.77342 &\\
& BMA & - & - & - & - & 6.33825 & \\
\hline
\multirow{3}*{pMTM}& MPM & 5.98 & 7.20 & 1.18 & 0.19519 & 7.04164 &\multirow{3}*{13.21}\\
& HPM & 6.63 & 7.10 & 1.73 & 0.25785 & 7.10383 &\\
& BMA & - & - & - & - & 6.57290 &\\
\hline
\multirow{3}*{pRNS (equal evaluations)}& MPM & 7.87 & 5.36 & 1.23 & 0.14951 & 6.06363 &\multirow{3}*{42.05}\\
& HPM & 8.31 & 5.37 & 1.68 & 0.18888 & 6.07039 &\\
& BMA & - & - & - & - & 5.98609 & \\
\hline
\multirow{3}*{pRNS (equal time)}& MPM & 7.71 & 5.70 & 1.41 & 0.17207 & 6.22964 &\multirow{3}*{13.87}\\
& HPM & 8.26 & 5.71 & 1.97 & 0.21550 & 6.26212 &\\
& BMA & - & - & - & - & \color{red}\textbf{6.15213} & \\
\hline
\multirow{3}*{pRNS (equal iterations)}& MPM & 4.38 & 9.89 & 2.27 & 0.51740 & 8.62847 &\multirow{3}*{0.27}\\
& HPM & 5.47 & 9.92 & 3.39 & 0.63110 & 8.84104 &\\
& BMA & - & - & - & - & 8.14917 & \\
\hline
\multicolumn{2}{c|}{EMVS}  & 7.37 & 6.58 & 1.95 & 0.25416 & 6.73390 & 18.91\\
\multicolumn{2}{c|}{Lasso} & 26.29 & 5.08 & 19.37 & 0.69542 & 6.51992 & 2.27\\
\multicolumn{2}{c|}{adaptive lasso} & 13.23 & 5.58 & 6.81 & 0.45352 & 6.28458 & 4.93\\
\multicolumn{2}{c|}{SCAD} & 21.48 & 6.29 & 15.77 & 0.70338 & 6.94900 & 4.61\\
\hline
\hline
\end{tabular}
\end{center}
\caption{Autoregressive correlation with $(n,p,p_0)=(100,1000,12)$}
\end{table}

\begin{table}[H]
\begin{center}
\scriptsize
\begin{tabular}{c c | c c c c c c}
\hline
\hline
\multicolumn{8}{c}{small correlation}  \\
& & \textit{Average model size} & \textit{FN} & \textit{FP} & \textit{FDR} & $\mathbf{\|\hat{\beta}-\beta\|_2}$  & \textit{Runtime}\\
\hline
\multirow{3}*{ada-pMTM}& MPM & 8.14 & 0.00 & 0.14 & 0.01511 & \textbf{0.36531} &\multirow{3}*{14.57}\\
& HPM & 8.09 & 0.00 & 0.09 & 0.01000 & \color{red}\textbf{0.34287} &\\
& BMA & - & - & - & - & 0.42152 & \\
\hline
\multirow{3}*{pMTM}& MPM & 8.08 & 0.12 & 0.20 & 0.02039 & 0.53798 &\multirow{3}*{15.10}\\
& HPM & 8.20 & 0.06 & 0.26 & 0.02622 & 0.48917 &\\
& BMA & - & - & - & - & 0.68383 &\\
\hline
\multirow{3}*{pRNS (equal evaluations)}& MPM & 8.35 & 0.00 & 0.35 & 0.03584 & 0.44451 &\multirow{3}*{40.72}\\
& HPM & 8.58 & 0.00 & 0.58 & 0.05574 & 0.49687 &\\
& BMA & - & - & - & - & 0.47022 & \\
\hline
\multirow{3}*{pRNS (equal time)}& MPM & 8.36 & 0.00 & 0.36 & 0.03717 & 0.46472 &\multirow{3}*{14.70}\\
& HPM & 8.48 & 0.00 & 0.48 & 0.04929 & 0.48844 &\\
& BMA & - & - & - & - & 0.48073 & \\
\hline
\multirow{3}*{pRNS (equal iterations)}& MPM & 1.90 & 6.32 & 0.22 & 0.10167 & 3.91957 &\multirow{3}*{0.23}\\
& HPM & 1.98 & 6.32 & 0.30 & 0.10804 & 3.97245 &\\
& BMA & - & - & - & - & 3.74767 & \\
\hline
\multicolumn{2}{c|}{EMVS}  & 8.09 & 0.05 & 0.14 & 0.01511 & 5.37749 & 12.57\\
\multicolumn{2}{c|}{Lasso} & 37.51 & 0.00 & 29.51 & 0.76589 & 1.61847 & 1.92\\
\multicolumn{2}{c|}{adaptive lasso} & 10.65 & 0.00 & 2.65 & 0.20397 & 0.50107 & 4.09\\
\multicolumn{2}{c|}{SCAD} & 14.54 & 0.00 & 6.54 & 0.27504 & 0.40595 & 3.50\\
\hline
\hline
\multicolumn{8}{c}{moderate correlation}  \\
& & \textit{Average model size} & \textit{FN} & \textit{FP} & \textit{FDR} & $\mathbf{\|\hat{\beta}-\beta\|_2}$  & \textit{Runtime}\\
\hline
\multirow{3}*{ada-pMTM}& MPM & 8.19 & 0.09 & 0.28 & 0.02999 & \textbf{0.47587} &\multirow{3}*{14.39}\\
& HPM & 8.19 & 0.09 & 0.28 & 0.02977 & \color{red}\textbf{0.47389} &\\
& BMA & - & - & - & - & 0.49938 & \\
\hline
\multirow{3}*{pMTM}& MPM & 7.76 & 0.85 & 0.61 & 0.08053 & 1.18489 &\multirow{3}*{14.88}\\
& HPM & 7.71 & 0.85 & 0.56 & 0.07571 & 1.06106 &\\
& BMA & - & - & - & - & 1.27868 &\\
\hline
\multirow{3}*{pRNS (equal evaluations)}& MPM & 8.47 & 0.00 & 0.47 & 0.04796 & 0.48850 &\multirow{3}*{38.86}\\
& HPM & 8.64 & 0.00 & 0.64 & 0.06330 & 0.53733 &\\
& BMA & - & - & - & - & 0.52055 & \\
\hline
\multirow{3}*{pRNS (equal time)}& MPM & 8.48 & 0.04 & 0.52 & 0.05139 & 0.52188 &\multirow{3}*{15.22}\\
& HPM & 8.65 & 0.04 & 0.69 & 0.06679 & 0.55481 &\\
& BMA & - & - & - & - & 0.54350 & \\
\hline
\multirow{3}*{pRNS (equal iterations)}& MPM & 2.29 & 6.25 & 0.54 & 0.20833 & 3.94940 &\multirow{3}*{0.22}\\
& HPM & 2.44 & 6.19 & 0.63 & 0.22275 & 3.94975 &\\
& BMA & - & - & - & - & 3.83011 & \\
\hline
\multicolumn{2}{c|}{EMVS}  & 7.60 & 1.22 & 0.82 & 0.11750 & 5.68208 & 13.15\\
\multicolumn{2}{c|}{Lasso} & 38.81 & 0.83 & 31.64 & 0.78922 & 2.88405 & 2.09\\
\multicolumn{2}{c|}{adaptive lasso} & 13.34 & 0.55 & 5.89 & 0.36971 & 1.46768 & 4.04\\
\multicolumn{2}{c|}{SCAD} & 18.02 & 0.20 & 10.22 & 0.38697 & 0.77203 & 4.60\\
\hline
\hline
\multicolumn{8}{c}{high correlation}  \\
& & \textit{Average model size} & \textit{FN} & \textit{FP} & \textit{FDR} & $\mathbf{\|\hat{\beta}-\beta\|_2}$  & \textit{Runtime}\\
\hline
\multirow{3}*{ada-pMTM}& MPM & 7.46 & 2.79 & 2.25 & 0.30544 & 3.01815 &\multirow{3}*{12.91}\\
& HPM & 7.42 & 2.78 & 2.20 & 0.29978 & 3.00770 &\\
& BMA & - & - & - & - & 3.02119 & \\
\hline
\multirow{3}*{pMTM}& MPM & 7.55 & 3.07 & 2.62 & 0.34795 & 3.27811 &\multirow{3}*{13.03}\\
& HPM & 7.50 & 3.05 & 2.55 & 0.34374 & 3.25323 &\\
& BMA & - & - & - & - & 3.23014 &\\
\hline
\multirow{3}*{pRNS (equal evaluations)}& MPM & 7.78 & 2.07 & 1.85 & 0.23894 & 2.50479 &\multirow{3}*{38.62}\\
& HPM & 7.89 & 2.06 & 1.95 & 0.24684 & 2.51739 &\\
& BMA & - & - & - & - & 2.50086 & \\
\hline
\multirow{3}*{pRNS (equal time)}& MPM & 7.46 & 2.52 & 1.98 & 0.26960 & 2.82655 &\multirow{3}*{13.88}\\
& HPM & 7.47 & 2.50 & 1.97 & 0.26696 & \textbf{2.80090} &\\
& BMA & - & - & - & - & \color{red}\textbf{2.79805} & \\
\hline
\multirow{3}*{pRNS (equal iterations)}& MPM & 3.79 & 5.96 & 1.75 & 0.47052 & 4.28314 &\multirow{3}*{0.27}\\
& HPM & 3.81 & 6.11 & 1.92 & 0.48923 & 4.37975 &\\
& BMA & - & - & - & - & 4.09175 & \\
\hline
\multicolumn{2}{c|}{EMVS}  & 7.99 & 2.85 & 2.84 & 0.35481 & 6.93585 & 10.36\\
\multicolumn{2}{c|}{Lasso} & 25.17 & 3.00 & 20.17 & 0.76385 & 3.67644 & 2.00\\
\multicolumn{2}{c|}{adaptive lasso} & 12.55 & 3.26 & 7.81 & 0.58423 & 3.84908 & 3.90\\
\multicolumn{2}{c|}{SCAD} & 15.76 & 3.01 & 10.77 & 0.58730 & 3.36388 & 3.72\\
\hline
\hline
\end{tabular}
\end{center}
\caption{Group structure with $(n,p,p_0)=(100,1000,8)$}
\end{table}
\subsection{A real data example}
The dataset [\cite{trindade2015}] contains electricity consumption of 370 clients from 2011 to 2014 without missing values. It records consumption every 15 minutes for every client in kW leading to 140,256 covariates. One question of interest is to predict the future electricity load in terms of previous usage. Therefore, the last column is treated as response and we preserve the top 10,000 covariates with large variances to reduce conditional number of the design matrix. Following \cite{wang2016}, the values are further scaled into [0,300]. The dataset is partitioned into training set with the first 200 clients and test set with remaining 170 clients. To reduce stochasticity, each method is run for 10 times except for \textit{EMVS} and \textit{SCAD}. We run \textit{ada-pMTM} using $M=1,000$ for 1000 iterations with first $20\%$ samples as burnin. All other settings keep same as Section \ref{4.2}. Running time, predictive MSE and model size are reported. As presented in Figure.\ref{fig3}, the predictive MSEs given by Bayesian model averaging and highest probability model of \textit{ada-pMTM} and \textit{pMTM} are smaller compared to other methods with competitive model sizes displayed in Figure.\ref{fig4}.
\begin{table}[H]
\begin{center}
\scriptsize
\begin{tabular}{c c c c c c c}
\hline
\textbf{ada-pMTM} & \textbf{pMTM} & \textbf{pRNS} & \textbf{Lasso} & \textbf{adaptive lasso} & \textbf{EMVS} & \textbf{SCAD} \\
\hline
120.56 & 116.43 & 117.82 & 72.39 & 141.91 & 65.76 & 65.36\\
\hline
\end{tabular}
\end{center}
\caption{Running time (secs) for the real example}
\label{table5}
\end{table}

\begin{figure}[H]
\centering
\includegraphics[width=10cm]{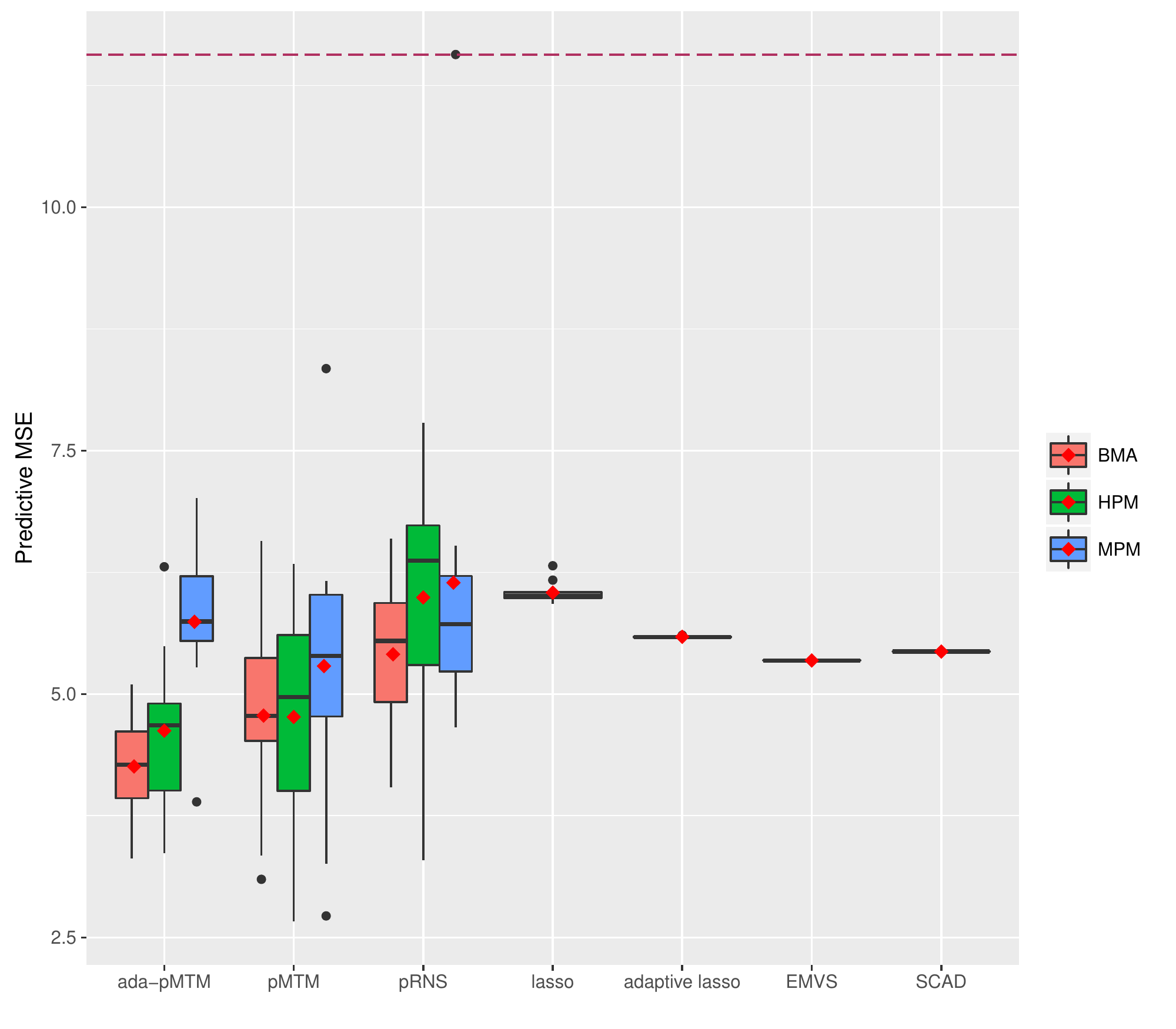}
\caption[]{Predictive MSE for different methods. The MSE using null model is marked as a purple dashed line. Red points represent means of boxes.}
\label{fig3}
\end{figure}
\begin{figure}[H]
\centering
\includegraphics[width=10cm]{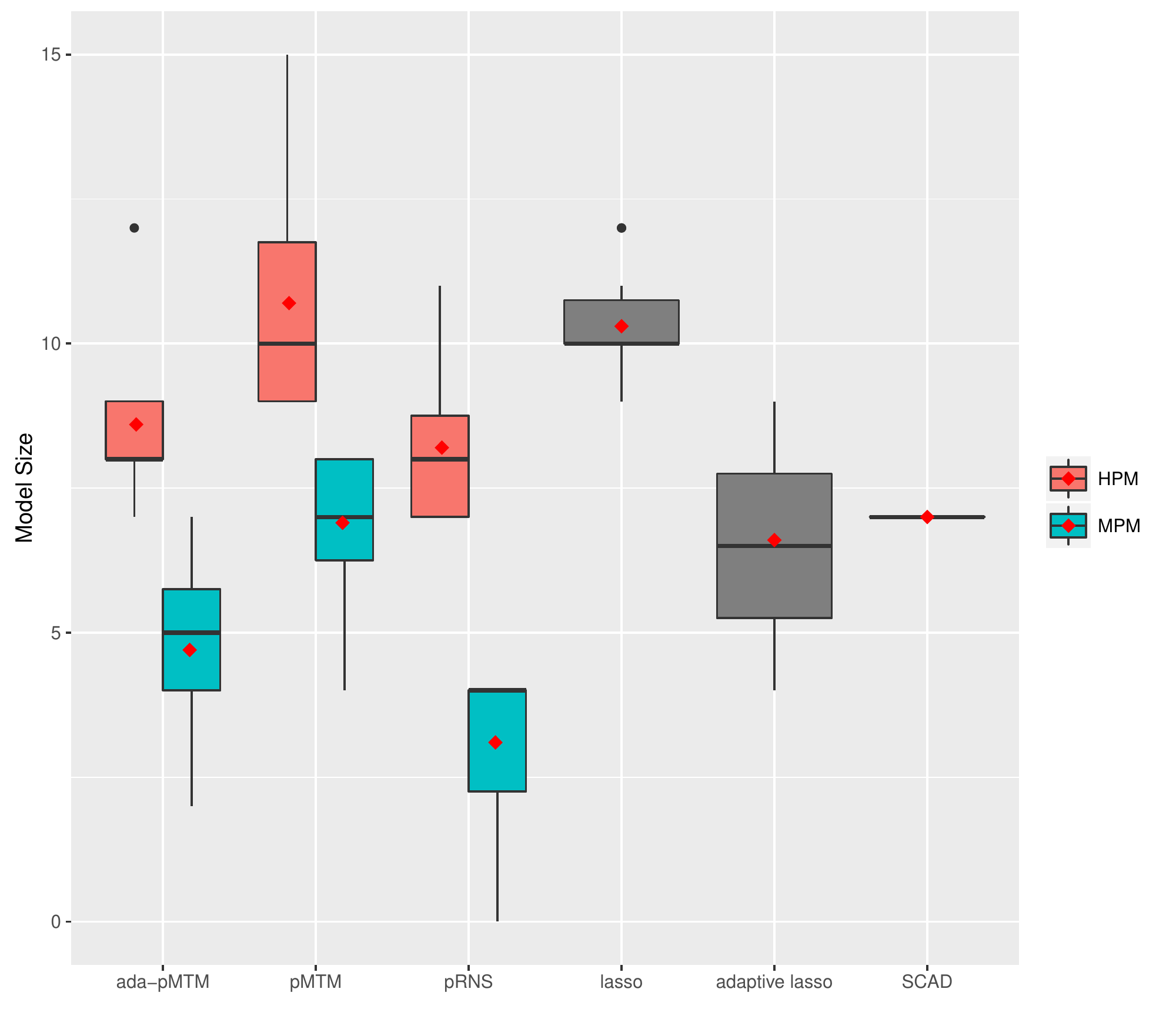}
\caption[]{Model size using different methods for the real data example}
\label{fig4}
\end{figure}
\subsection{Computational efficiency}\label{comeff}
Scalability is another attractiveness of the proposed algorithms. In this section, we compare the computational efficiency of the proposed algorithms run with different number of cores on a simulated dataset with independent design in Section \ref{4.2} for $n=10^3$ and $p=2\times10^4$. \texttt{apply} function is used for \textit{pMTM} and \textit{ada-pMTM} with single core when evaluating marginal likelihoods. \textit{pMTM, ada-pMTM-4, 8} represent \textit{pMTM} or \textit{ada-pMTM} run on 4 or 8 clusters. For parallelization, datasets are first distributed to multiple clusters and then \texttt{parLapply} in \texttt{parallel} package is used. All algorithms are implemented on the same 10 synthetic datasets at each value of $M$. A graph of the mean numbers of evaluations of marginal likelihood within 10 seconds against $M/p$ is provided in Figure.\ref{fig6}. 

The line for \textit{pRNS} is a constant since the algorithm does not involve $M$. \textit{ada-pMTM} needs to update scores for predictors and hence it evaluates less marginal likelihoods than \textit{pMTM}. 2 communications are required at each iteration and hence parallelization with 4 or 8 clusters when $M=p/5$ is not beneficial. When $M$ becomes larger, computing time overwhelms communication time resulting in the dominance of algorithms implemented with 8 clusters. 
\begin{figure}[H]
\centering
\includegraphics[width=10cm]{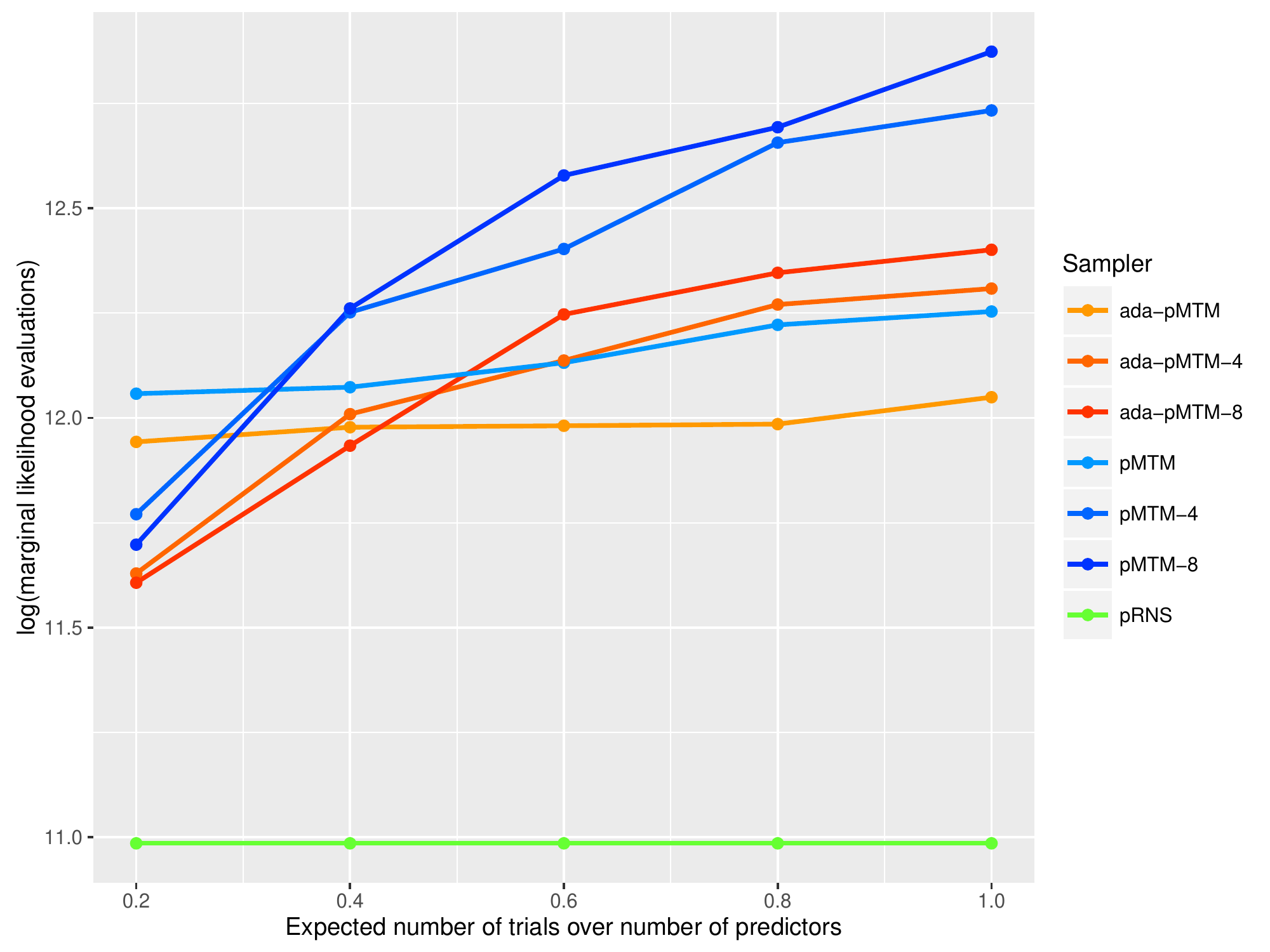}
\caption[]{Logarithm of the mean number of marginal likelihood evaluations within 10 seconds}
\label{fig6}
\end{figure}

\section{Discussion}
\label{sec5}
We propose a paired-move multiple-try Metropolis MCMC sampler for Bayesian variable selection. Extensive simulation studies demonstrate the effectiveness of \textit{pMTM} especially for ``large $p$ small $n$'' scenario. Efficient model space exploration with less computational cost is achieved by incorporating the paired-move and multiple-try strategies. Comparing to \textit{SSS}, a more flexible computational budget can be determined manually based on data and purpose instead of considering all neighborhoods. In this work, the expected computational budget $M$ is specified as $p/10$. However, the optimal choice of $M$ is still not fully explored. Intuitively, the optimal $M$ may depend on dimensions and correlation structure of the design matrix.

Reproducibility is a key issue in scientific research [\cite{peng2011}; \cite{collins2014}; \cite{open2015}]. Research based on statistical computations is expected to be able to be replicated. In the context of inference using MCMC techniques, both of the following two elements are required for reproducibility:
\begin{enumerate}
\item convergence of the Markov chain: To ensure the samples are indeed drawn from the target distribution, we require the chain nearly converging to the equilibrium.
\item enough posterior samples: Bayesian inference is mostly based on posterior samples. Therefore, enough posterior samples drawn from a converged chain are required to make accurate inference.
\end{enumerate}
Considering running the proposed algorithms under fixed running time, chains produced by \textit{pMTM} and \textit{ada-pMTM} can rapidly converge to equilibrium with a small number of posterior samples while \textit{pRNS} can generate a large number of samples but may be stuck in some local modes. Therefore, implementing each of these algorithms in a short period of time may fail to simultaneously satisfy the two requirements. A hybrid algorithm, combining \textit{pRNS} and \textit{ada-pMTM}, that take advantages of both is worthwhile developing. 
 
To facilitate the application of our method to even huge datasets, one may further accelerate \textit{pMTM} by subsampling [\cite{DBLP:conf/icml/BalanCW14}; \cite{quiroz2015}] which is randomly selecting a mini-batch of samples at each iteration for computing marginal likelihoods. Another possible approach is to partition the design matrix first either using sample space partitioning [\cite{wang2014}] or feature space partitioning [\cite{wang2016}] and then apply \textit{pMTM} on each subset of data. 
%\begin{supplement}
%\sname{Supplement A}\label{suppA} 
%\stitle{Title of the Supplement A}
%\slink[url]{http://www.some-url-address.org/dowload/0000.zip}
%\sdescription{Add description for supplement material.}
%\end{supplement}

\bibliographystyle{ba}
\bibliography{sample}

%\begin{acknowledgement}
%And this is an acknowledgements section with a heading that was produced by the
%$\backslash$section* command. Thank you all for helping me writing this
%\LaTeX\ sample file. See \ref{suppA} for the supplementary material example.
%\end{acknowledgement}

\end{document}